\documentclass[envcountsame]{llncs}

\usepackage{amsmath,amssymb,amsfonts}
\usepackage{thm-restate}
\usepackage{cleveref}

\usepackage{tikz}
\usetikzlibrary{arrows,calc,automata}
\tikzset{LMC style/.style={>=stealth',every edge/.append style={thick},every state/.style={minimum size=18,inner sep=0}}}

\def\eqdef{\stackrel{\text{def}}{=}}

\newcommand{\A}{\mathcal{A}}
\newcommand{\AF}{\overrightarrow{\mathcal{A}}}

\newcommand{\alphaF}{\overrightarrow{\alpha}}
\newcommand{\B}{\mathcal{B}}
\newcommand{\Bin}{\mathbb{B}}

\newcommand{\F}{\mathbb{F}}

\newcommand{\kA}{\mathcal{A}_k}
\newcommand{\MF}{\overrightarrow{M}}
\newcommand{\N}{\mathbb{N}}
\newcommand{\Q}{\mathbb{Q}}
\newcommand{\QF}{\overrightarrow{Q}}
\newcommand{\R}{\mathbb{R}}
\DeclareMathOperator{\rank}{rank}
\newcommand{\VV}{\mathsf{V}}
\newcommand{\Z}{\mathbb{Z}}

\DeclareMathOperator{\prob}{\mathit{Pr}}
\DeclareMathOperator{\expect}{\mathbb{E}}
\newcommand{\paths}{\mathit{Paths}}
\DeclareMathOperator{\kdis}{\mathnormal{k}-dis}
\newcommand{\pow}{\mathcal{P}}

\begin{document}

\title
{Image-Binary Automata}

\author{Stefan Kiefer \and
	Cas Widdershoven}
	
\institute{Department of Computer Science,
	University of Oxford}

\authorrunning{S.~Kiefer and C.~Widdershoven} 

\sloppy


\maketitle              

\begin{abstract}
We introduce a certain restriction of weighted automata over the rationals, called \emph{image-binary automata}. We show that such automata accept the regular languages, can be exponentially more succinct than corresponding NFAs, and allow for polynomial complementation, union, and intersection.
This compares favourably with unambiguous automata whose complementation requires a superpolynomial state blowup.
We also study an infinite-word version, \emph{image-binary B\"uchi automata}.
We show that such automata are amenable to probabilistic model checking, similarly to unambiguous B\"uchi automata.
We provide algorithms to translate $k$-ambiguous B\"uchi automata to image-binary B\"uchi automata, leading to model-checking algorithms with optimal computational complexity.
\end{abstract}

\section{Introduction} \label{sec-intro}

A weighted automaton assigns weights to words; i.e., it defines mappings of the form $f : \Sigma^* \to D$, where $D$ is some domain of weights.
Weighted automata are well-studied. Many variations have been discussed, such as max-plus automata~\cite{droste09} and probabilistic automata \cite{rabin63,paz14}, both over finite words and over infinite words, in the latter case often combined with $\omega$-valuation monoids~\cite{droste11}. However, it has been shown that many natural questions are undecidable for many kinds of weighted automata~\cite{fijalkow17,almagor20}, including inclusion and equivalence.
These problems become decidable for finitely ambiguous weighted automata~\cite{filiot14}.

In this paper we consider only numerical weights, where $D$ is a subfield of the reals.
A \emph{language} $L \subseteq \Sigma^*$ can be identified with its characteristic function $\chi_L : \Sigma^* \to \{0,1\}$.
We explore weighted automata that encode characteristic functions of languages, i.e., weighted automata that map each word either to $0$ or to~$1$.
We call such automata \emph{image-binary finite automata (IFAs)} and view them as acceptors of languages $L \subseteq \Sigma^*$.
We do not require, however, that individual transitions have weight $0$ or~$1$.
This makes IFAs a ``semantic'' class: it may not be obvious from the transition weights whether a given weighted automaton over, say, the rationals is image-binary.
However, we will see that it can be checked efficiently whether a given $\Q$-weighted automaton is an IFA (\Cref{thm-check-image-binary}).

An immediate question is on the expressive power of IFAs.
Deterministic finite automata (DFAs) can be viewed as IFAs.
On the other hand, in \Cref{sub-regularity} we show that all languages accepted by IFAs are regular.
It follows that IFAs accept exactly the regular languages.
Moreover, IFAs are efficiently closed under Boolean operations; i.e., given two IFAs that accept $L_1, L_2$, respectively, one can compute in polynomial time IFAs accepting $L_1 \cup L_2$, $L_1 \cap L_2$, and $\Sigma^* \setminus L_1$.

The latter feature, efficient closure under complement, might be viewed as a key advantage of IFAs over unambiguous finite automata (UFAs).
UFAs are nondeterministic finite automata (NFAs) such that every word has either zero or one accepting runs.
UFAs can be viewed as a special case of IFAs.
Whereas we show that IFAs can be complemented in polynomial time, UFAs are known to be not polynomially closed under complement~\cite{Raskin18}.

The next question is then on the succinctness of IFAs, and on the complexity of converting other types of finite automata to IFAs and vice versa.
We study such questions in \Cref{sub-succinctness}.
In \Cref{sub-mod-2-automata}, we also study the relationship of IFAs to \emph{mod-2 multiplicity automata}, which are weighted automata over~$\mathit{GF}(2)$, the field $\{0,1\}$ where $1+1=0$.
Such automata~\cite{angluin2020} share various features with IFAs, in particular efficient closure under complement.

In the second part of the paper we put IFAs ``to work''.
Specifically, we consider an infinite-word version, which we call \emph{image-binary B\"uchi automata (IBAs)}.
Following the theme that image-binary automata naturally generalise and relax unambiguous automata, we show that IBAs can be used for model checking Markov chains in essentially the same way as \emph{unambiguous B\"uchi automata (UBAs)}~\cite{cav16full}. 
Specifically, we show in \Cref{sub-mc-iba} that given an IBA and a Markov chain, one can compute in NC (hence in polynomial time) the probability that a random word produced by the Markov chain is accepted by the IBA.

It was shown in~\cite{loding18} that a nondeterministic B\"uchi automaton (NBA) with $n$~states can be converted to an NBA with at most $3^n$~states whose ambiguity is bounded by~$n$.
Known conversions from NBAs to UBAs have a state blowup of roughly $n^n$, see, e.g., \cite{KarmarkarJC13}.
We show in~\Cref{sub-model-kaba} that NBAs with logarithmic ambiguity (as produced by the construction from~\cite{loding18}) can be converted to IBAs in polylogarithmic space.
This suggests that in order to translate NBAs into an automaton model suitable for probabilistic model checking (such as IBAs), it is reasonable to first employ the partial disambiguation procedure from~\cite{loding18} (which does most of the work).
More specifically, by combining the partial disambiguation procedure from~\cite{loding18} with our translation to an IBA, we obtain a PSPACE transducer (i.e., a Turing machine whose work tape is polynomially bounded) that translates an NBA into an IBA.
For example, combining that with the mentioned probabilistic model checking procedure for IBAs we obtain an (optimal) PSPACE procedure for model checking Markov chains against NBA specifications.

\section{Image-Binary Finite Automata} \label{sec-IFA}

\subsection{Definitions}\label{sub-ifa-defs}

Let $\F$ be one of the fields $\Q$ or~$\R$ (with ordinary addition and multiplication).
An \emph{$\F$-weighted automaton} $\A = (Q, \Sigma, M, \alpha, \eta)$ consists of a set of states~$Q$,
 a finite alphabet~$\Sigma$,
 a map $M : \Sigma \to \F^{Q \times Q}$, an initial (row) vector $\alpha \in \F^Q$, and a final (column) vector $\eta \in \F^Q$.
Extend $M$ to $\Sigma^*$ by setting $M(a_1 \cdots a_k) \eqdef M(a_1) \cdots M(a_k)$.
The \emph{language}~$L_\A$ of an automaton~$\A$ is the map $L_\A : \Sigma^* \to \F$ with $L_\A(w) = \alpha M(w) \eta$.
Automata $\A, \B$ over the same alphabet~$\Sigma$ are said to be \emph{equivalent} if $L_\A = L_\B$.

Let $\A = (Q, \Sigma, M, \alpha, \eta)$ be a $\Q$-weighted automaton.
We call $\A$ an \emph{image-binary (weighted) finite automaton (IFA)} if $L_\A(\Sigma^*) \subseteq \{0,1\}$, i.e., $L_\A(w) \in \{0,1\}$ holds for all $w \in \Sigma^*$.
An $\R$-IFA is defined like an IFA, but with $\Q$ replaced by~$\R$.
An ($\R$)-IFA~$\A$ defines a language $L(\A) := \{w \in \Sigma^* \mid L_\A(w) = 1\}$.
Note that we call both $L_\A$ and $L(\A)$ the language of~$A$; strictly speaking, the former is the characteristic function of the latter.

If an IFA $\A = (Q, \Sigma, M, \alpha, \eta)$ is such that $\alpha \in \{0,1\}^Q$ and $\eta \in \{0,1\}^Q$ and $M(a) \in \{0,1\}^{Q \times Q}$ for all $a \in \Sigma$, then $\A$ is called an \emph{unambiguous finite automaton (UFA)}.
Note that this definition of a UFA is essentially equivalent to the classical one, which says that a UFA is an NFA (nondeterministic finite automaton) where each word has at most~$1$ accepting run.
Similarly, a \emph{deterministic finite automaton (DFA)} is essentially a special case of a UFA, and hence of an IFA.

\begin{example} \label{ex-UFA-IFA}
\Cref{fig-ex-IFA} shows an IFA and a UFA in a graphical notation.
Formally, the IFA on the left is $\A = (Q_\A, \Sigma, M_\A, \alpha_\A, \eta_\A)$ with $Q = \{1,2,3\}$ and $\Sigma = \{a,b\}$ and
\[
M_\A(a) \ = \ \begin{pmatrix}
-1 & 1 & 0 \\
0  & 0 & 1 \\
0  & 0 & 1
\end{pmatrix}
\quad \text{and} \quad
M_\A(b) \ = \ \begin{pmatrix}
0 & 0 & 0 \\
0 & 0 & 0 \\
0 & 0 & 1
\end{pmatrix}
\]
and $\alpha_\A = \begin{pmatrix} 1 & 0 & 0 \end{pmatrix}$ and $\eta_\A = \begin{pmatrix} 0 & 0 & 1 \end{pmatrix}^T$.
Both automata recognise the language of words that start in an even (positive) number of $a$s.
\end{example}

\begin{figure}
\begin{center}
\resizebox{0.8\textwidth}{!}{%
\begin{tikzpicture}[scale=2.5,LMC style]
\node[state] (1) at (0,0) {$1$};
\node[state] (2) at (1,0) {$2$};
\node[state,accepting] (3) at (2,0) {$3$};
\path[->] (-0.5,0) edge (1);
\path[->] (1) edge [loop,out=70,in=110,looseness=10] node[above] {$-1 a$} (1);
\path[->] (1) edge node[above] {$a$} (2);
\path[->] (2) edge node[above] {$a$} (3);
\path[->] (3) edge [loop,out=70,in=110,looseness=10] node[above] {$a,b$} (3);

\node (a) at (0.75,-1) {(a)};

\node[state] (q1) at (3.5,0.5) {$1$};
\node[state] (q2) at (4.5,0.5) {$2$};
\node[state,accepting] (q3) at (4,-0.2) {$3$};
\path[->] (3,0.5) edge (q1);
\path[->] (q1) edge[bend left=20] node[above] {$a$} (q2);
\path[->] (q2) edge[bend left=20] node[above] {$a$} (q1);
\path[->] (q2) edge[bend left=20] node[right] {$a$} (q3);
\path[->] (q3) edge[bend left=20] node[left] {$b$} (q2);
\path[->] (q3) edge node[left,pos=0.4] {$b$} (q1);
\path[->] (q3) edge [loop,out=250,in=290,looseness=10] node[below] {$b$} (q3);

\node (b) at (4,-1) {(b)};
\end{tikzpicture}
}
\end{center}
\caption{The IFA in (a) is a forward conjugate of, and hence equivalent to, the UFA in (b). Unless indicated otherwise, edges in (a) have weight 1.}
\label{fig-ex-IFA}
\end{figure}
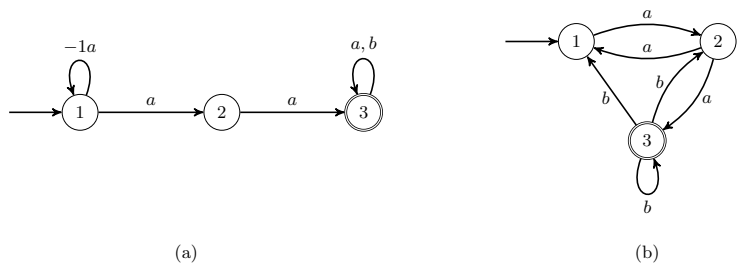

Let $\A = (Q, \Sigma, M, \alpha, \eta)$ be an $\R$-weighted automaton.
We call $\AF := (\QF, \Sigma, \MF, \alphaF, F \eta)$ a \emph{forward conjugate} of~$\A$ with base~$F$ if $F \in \R^{\QF \times Q}$ and $F M(a) = \MF(a) F$ for all $a \in \Sigma$ and $\alpha = \alphaF F$.
Such $\A$ and $\AF$ are equivalent: indeed, let $w \in \Sigma^*$; by induction we have $F M(w) = \MF(w) F$ and hence
\[
 L_\A(w) \ = \ \alpha M(w) \eta \ = \ \alphaF F M(w) \eta \ = \ \alphaF \MF(w) F \eta \ = \ L_{\AF}(w)\,.
\]
A backward conjugate can be defined analogously.
\begin{example} \label{ex-forward-conjugate}
The IFA~$\A$ on the left of \Cref{fig-ex-IFA} is a forward conjugate of the UFA on the right with base
\[
F \ = \ \begin{pmatrix}
1 & 0 & 0 \\
1 & 1 & 0 \\
1 & 1 & 1
\end{pmatrix}\,.
\]
Indeed, we have $\begin{pmatrix} 1 & 0 & 0 \end{pmatrix} F = \begin{pmatrix} 1 & 0 & 0 \end{pmatrix}$ and $\begin{pmatrix} 0 & 0 & 1\end{pmatrix}^T = F \begin{pmatrix} 0 & 0 & 1 \end{pmatrix}^T$, where $\vec{v}^T$ denotes the transpose of a vector $\vec{v}$, and
\[
F \begin{pmatrix} 0 & 1 & 0 \\ 1 & 0 & 1 \\ 0 & 0 & 0 \end{pmatrix} \, = \,
\begin{pmatrix}
0 & 1 & 0 \\ 1 & 1 & 1 \\ 1 & 1 & 1
\end{pmatrix}
\, = \,
\begin{pmatrix}
-1 & 1 & 0 \\
0  & 0 & 1 \\
0  & 0 & 1
\end{pmatrix}
F
\quad
\text{and}
\quad
F \begin{pmatrix} 0 & 0 & 0 \\ 0 & 0 & 0 \\ 1 & 1 & 1\end{pmatrix} \, = \, \begin{pmatrix} 0 & 0 & 0 \\ 0 & 0 & 0 \\ 0 & 0 & 1\end{pmatrix}
F.
\]
\end{example}

For some proofs we need the following definition.
Let $L : \Sigma^* \to \F$, where $\F$ is any field.
Then the \emph{Hankel matrix} of~$L$ is the infinite matrix $H_L \in \F^{\Sigma^* \times \Sigma^*}$ with $H_L[x,y] = L(x y)$.
It was shown by Carlyle and Paz~\cite{carlyle1971realizations} and Fliess~\cite{fliess} that the rank of~$H_L$ is equal to the number of states of the minimal (in number of states) $\F$-weighted automaton~$\A$ with $L_\A = L$.
\begin{proposition}[\cite{carlyle1971realizations,fliess}] \label{prop-Hankel}
Let $L$ be an $\F$-weighted regular language, i.e. a function $L  : \Sigma^* \to \F$ that can be represented by an $\F$-weighted automaton.
Let $\A = (Q, \Sigma, M, \alpha, \eta)$ be a minimal $\F$-weighted automaton such that $L_\A = L$.
Then $\rank H_{L} = |Q|$.
\end{proposition}

\subsection{Regularity} \label{sub-regularity}

Since a DFA is an IFA, for each regular language there is an IFA that defines it.
Conversely, we show that the language of an IFA is regular:
\begin{theorem} \label{thm-regularity}
Let $\A = (Q, \Sigma, M, \alpha, \eta)$ be an $\R$-IFA.
Then $L(\A)$ is regular, and there is a DFA $\B$ with at most $2^{|Q|}$ states and $L(\A) = L(\B)$.
\end{theorem}

In order to do so, we will need some auxiliary lemmas.
View $\Z_2 = \{0,1\}$ as the field with two elements.
In the proof of \Cref{thm-regularity} we consider vector spaces \emph{over}~$\Z_2$, i.e., where the scalars are from~$\Z_2$.
In particular, we will argue with the vector space $\Z_2^\N \cong \Z_2^{\Sigma^*}$ over~$\Z_2$.
We first show:
\begin{lemma} \label{lem-Z2}
Let $V$ be a set of $n$ vectors.
Consider the vector space $\langle V \rangle$ spanned by~$V$ over~$\Z_2$.
Then $|\langle V \rangle| \le 2^n$.
\end{lemma}
\begin{proof}
Let $V = \{v_1, \ldots, v_n\}$.
Then $\langle V \rangle = \{\sum_{i=1}^n \lambda_i v_i \mid \lambda_i \in \Z_2\}$.
\qed
\end{proof}
\begin{corollary} \label{cor-Z2}
Let $\VV$ be a vector space over~$\Z_2$.
For any $n \in \N$, if $\dim \VV \le n$ then $|\VV| \le 2^n$.
\end{corollary}

The following two lemmas show that if an $\R$-weighted automaton is image-binary, then the rank over $\R$ of its Hankel matrix $H$ is at least the rank of $H$ over $\Z_2$.
\begin{lemma} \label{lem-R-Q}
Let $V \subseteq \{0,1\}^{\N}$ be a set of vectors.
If $V$ is linearly dependent over~$\R$ then $V$ is linearly dependent over~$\Q$.
Hence $\dim \langle V \rangle_{\Q} \le \dim \langle V \rangle_{\R}$.
\end{lemma}
\begin{proof}
Let $V$ be linearly dependent (over~$\R$), and let $n = \dim \langle V \rangle$.
Then there are $v_0, v_1, \ldots, v_n \in V$ such that $V' := \{v_1, \ldots, v_n\}$ is linearly independent and $v_0 \not\in V'$ but $v_0 \in \langle V' \rangle$.
So there are unique $\lambda_1, \ldots, \lambda_n \in \R$ with $v_0 = \sum_{i=1}^n \lambda_i v_i$.
It suffices to show that $\lambda_1, \ldots, \lambda_n \in \Q$.
Since $V'$ is linearly independent, there exists $J \subseteq \N$ with $|J| = n$ such that $\{v_1[J], \ldots, v_n[J]\}$ is linearly independent, where $v_i[J] \in \{0,1\}^n$ is the restriction of $v_i$ to the entries indexed by~$J$.
Hence the $\lambda_i$ are the unique solution of a linear system of equations
\begin{align*}
v_0[j] \ &=\ \sum_{i=1}^n \lambda_i v_i[j] \qquad \text{ for all $j \in J$.}
\end{align*}
Linear systems of equations with rational coefficients and constants have rational solutions.
Hence $\lambda_1, \ldots, \lambda_n \in \Q$.
\qed
\end{proof}
\begin{lemma} \label{lem-Z2-R}
Let $V \subseteq \{0,1\}^{\N}$ be a set of vectors.
If $V$ is linearly dependent over~$\Q$ then $V$ is linearly dependent over~$\Z_2$.
Hence $\dim \langle V \rangle_{\Z_2} \le \dim \langle V \rangle_{\Q}$.
\end{lemma}
\begin{proof}
Let $V = \{v_1, \ldots, v_n\}$ and $\sum_{i=1}^n \lambda_i v_i = \vec0$, where $\lambda_i \in \Q$ are not all zero.
By multiplying all $\lambda_i$ with a common denominator, we can assume without loss of generality that $\lambda_i \in \Z$ where $\lambda_i$ are not all zero.
By dividing all $\lambda_i$ with the largest power of~$2$ that divides all~$\lambda_i$, we can assume without loss of generality that the $\lambda_i$ are not all even.
In the equation $\sum_{i=1}^n \lambda_i v_i = \vec0$, by regarding every $\lambda_i$ and every entry of every $v_i$ modulo~$2$, we have $\sum_{i=1}^n \lambda_i v_i = \vec0$ over~$\Z_2$, i.e., $V$ is linearly dependent over~$\Z_2$.
\qed
\end{proof}

Hence we can prove~\Cref{thm-regularity}:

\begin{proof}[of \Cref{thm-regularity}]
Write $n = |Q|$.
Let $H$ be the Hankel matrix of $L_\A$.
We have $H \in \{0,1\}^{\Sigma^* \times \Sigma^*}$.
By \cref{prop-Hankel} we have $\rank H \le n$, where the rank is over~$\R$.
By \cref{lem-R-Q,lem-Z2-R} it follows that $\rank H \le n$, where the rank is over~$\Z_2$.
By \cref{cor-Z2} it follows that $H$ has at most $2^n$ different rows, say $H[w_1,\cdot], \ldots, H[w_\ell, \cdot]$ with $\ell \le 2^n$.
So every word is Myhill-Nerode equivalent to a word in $\{w_1, \ldots, w_\ell\}$.
Thus, there is an equivalent DFA with $\ell$ states.

Explicitly, the following DFA $\B = (Q', \Sigma, \delta, q_0, F)$ is equivalent to~$\A$:
\begin{align*}
 Q'   \ &=\ \{H[w_1, \cdot], \ldots, H[w_\ell,\cdot]\} \\
 \delta(H[w_i, \cdot], a) \ &=\ H[w_i a, \cdot] \\
 q_0 \ &=\ H[\varepsilon,\cdot] \\
 F   \ &=\ \{H[w_i, \cdot] \mid 1 \le i \le \ell,\ H[w_i, \varepsilon] = 1\}\,. \tag*{\qed}
\end{align*}
\end{proof}

\subsection{Boolean Operations and Checking Image-Binariness} \label{sub-boolean-ops}

IFAs are \emph{polynomially closed} under all boolean operations, by which we mean:
\begin{theorem} \label{thm-boolean-ops}
Let $\A_1, \A_2$ be IFAs over~$\Sigma$.
One can compute in polynomial time IFAs $\B_\neg, \B_\cap, \B_\cup$ with $L(\B_\neg) = \Sigma^* \setminus L(\A_1)$ and $L(\B_\cap) = L(\A_1) \cap L(\A_2)$ and $L(\B_\cup) = L(\A_1) \cup L(\A_2)$.
\end{theorem}
By De Morgan's laws it suffices to construct $\B_\neg$ and $\B_\cap$.
Since $\B_\neg$ and~$\B_\cap$ need to satisfy only $L_{\B_\neg} = 1 + (-L_{\A_1})$ (where $1 : \Sigma^* \to \{1\}$ denotes the constant~$1$ function) and $L_{\B_\cap} = L_{\A_1} \cdot L_{\A_2}$, it suffices to know that $\Q$-weighted automata are polynomially closed under negation and pointwise addition and multiplication:
\begin{proposition}[see, e.g., {\cite[Chapter~1]{BerstelReutenauer88}}] \label{prop-weighted-closure}
Let $\A_1, \A_2$ be $\Q$-weighted automata.
One can compute in polynomial time $\Q$-weighted automata $\B_-, \B_+, \B_\times$ with $L_{\B_-} = - L_{\A_1}$ and $L_{\B_+} = L_{\A_1} + L_{\A_2}$ and $L_{\B_\times} = L_{\A_1} \cdot L_{\A_2}$.
\end{proposition}
\begin{proof}
For $i \in \{1,2\}$ let $\A_i = (Q_i, \Sigma, M_i, \alpha_i, \eta_i)$.
For $\B_-$, replace $\alpha_1$ by $-\alpha_1$.
For $\B_+$, assume $Q_1 \cap Q_2 = \emptyset$, take $Q = Q_1 \cup Q_2$, and $M(a) = \begin{pmatrix} M_1(a) & 0 \\ 0 & M_2(a) \end{pmatrix}$ for all $a \in \Sigma$, and $\alpha = \begin{pmatrix} \alpha_1 & \alpha_2 \end{pmatrix}$, and $\eta = \begin{pmatrix} \eta_1 \\ \eta_2\end{pmatrix}$.
For $\B_\times$, take $Q = Q_1 \times Q_2$, and $M(a) = M_1(a) \otimes M_2(a)$ for all $a \in \Sigma$, and $\alpha = \alpha_1 \otimes \alpha_2$, and $\eta = \eta_1 \otimes \eta_2$, where $\mathord{\otimes}$ stands for the Kronecker product.
It follows from the mixed-product property of $\mathord{\otimes}$ (i.e., $(A \otimes B) (C \otimes D) = (A C) \otimes (B D)$) that indeed $L_{\B_\times} = L_{\A_1} \cdot L_{\A_2}$.
\qed
\end{proof}

While DFAs are also polynomially closed under complement (switch accepting and non-accepting states), NFAs and UFAs are not.
For NFAs, it was shown in~\cite{HolzerK03} that the (worst-case) blowup in the number~$n$ of states is $\Theta(2^n)$.
For UFAs, it was shown recently:
\begin{proposition}[\cite{Raskin18}] \label{prop-Raskin18}
For any $n \in \N$ there exists a unary (i.e., on an alphabet $\Sigma$ with $|\Sigma|=1$) UFA $\A_n$ with $n$~states such that any NFA for the complement language has at least $n^{(\log \log \log n)^{\Theta(1)}}$ states.
\end{proposition}
This super-polynomial blowup (even for unary alphabet and even if the output automaton is allowed to be ambiguous) refuted a conjecture that it may be possible to complement UFAs with a polynomial blowup~\cite{Colcombet15}.
An upper bound (for general alphabets and requiring the output to be a UFA) of $O(2^{0.79 n})$ was shown in~\cite{JirasekJS18}; see also~\cite{indz2021} for an (unpublished) improvement.

The authors believe \Cref{prop-Raskin18} shows the strength of \Cref{thm-boolean-ops}:
while UFAs cannot be complemented efficiently, the more general IFAs are polynomially closed under all boolean operations.

\Cref{prop-weighted-closure} can be used to show: 
\begin{theorem} \label{thm-check-image-binary}
Given a $\Q$-weighted automaton, one can check in polynomial time if it is an IFA.
\end{theorem}
\begin{proof}
Let $\A$ be a $\Q$-weighted automaton.
By \Cref{prop-weighted-closure} one can compute in polynomial time a $\Q$-weighted automaton~$\B$ with $L_\B = L_\A \cdot L_\A$ (pointwise multiplication).
Then $\A$ is an IFA if and only if $\A$ and~$\B$ are equivalent.
Equivalence of $\Q$-weighted automata can be checked in polynomial time, see~\cite{schut61,Tzeng96}.
\qed
\end{proof}

\subsection{Succinctness} \label{sub-succinctness}

It is known that UFAs can be exponentially more succinct than DFAs: for each $n \in \N$ with $n \ge 3$ there is a UFA with $n$ states such that the smallest equivalent DFA has $2^n$ states, see~\cite[Theorem~1]{Leung05}.
Since UFAs are IFAs, \Cref{thm-regularity} is optimal:
\begin{corollary} \label{cor-IFA2DFA}
For converting IFAs to DFAs, a state blowup of $2^n$ is sufficient and necessary.
\end{corollary}

It is also known from~\cite{Leung05} that converting NFAs to UFAs can require $2^n - 1$ states.
The argument carries over to IFAs:
\begin{proposition} \label{prop-NFA2IFA}
For converting NFAs to IFAs, a state blowup of $\Theta(2^n)$ is sufficient and necessary.
\end{proposition}
\begin{proof}
Sufficiency is clear via the subset construction.
For necessity, Leung~\cite[Theorem~3]{Leung05} considers for each $n \ge 3$ an NFA (even an MDFA, which is a DFA with multiple initial states) such that its Hankel matrix has rank $2^n - 1$, which he proved in~\cite{Leung98}.
He then invokes an analogue of \Cref{prop-Hankel}, due to Schmidt~\cite{Schmidt78}, to show that any equivalent UFA needs at least $2^n - 1$ states.
But by \Cref{prop-Hankel} this holds also for IFAs.
\qed
\end{proof}

It follows from \Cref{thm-boolean-ops} and \Cref{prop-Raskin18} that IFAs cannot be converted to NFAs in polynomial time:
\begin{proposition} \label{prop-IFA2NFA}
Converting IFAs to NFAs requires a super-polynomial state blowup.
\end{proposition}
\begin{proof}
Let $n \in \N$, and let $\A_n$ be the UFA from \cref{prop-Raskin18}.
By \cref{thm-boolean-ops} there is a polynomial-size IFA, say~$\B_n$, for the complement of~$L(\A_n)$.
If converting IFAs to NFAs required only a polynomial state blowup, there would exist an NFA, say~$\B_n'$, with $L(\B_n') = L(\B_n) = \Sigma^* \setminus L(\A_n)$, of size polynomial in~$\A_n$, contradicting \cref{prop-Raskin18}.
\qed
\end{proof}

\subsection{Mod-2-Multiplicity Automata}\label{sub-mod-2-automata}

We compare IFAs with the mod-2-multiplicity automata (mod-2-MAs) as introduced in~\cite{angluin2020}, which are weighted automata over the field $\Z_2$.
Given a mod-2-MA $\mathcal{A}$ and a word $w$, $w$ is accepted iff $\mathcal{A}(w) = 1$.
Like with IFAs, mod-2-MAs are exponentially more succinct than DFAs~\cite[Lemma~6]{angluin2020}.
Converting NFAs to mod-2-MAs requires a super-polynomial state blowup~\cite[Lemma~10]{angluin2020} while (under the assumption that there are infinitely many Mersenne primes) converting mod-2-MAs to NFAs requires an exponential blowup~\cite[Lemma~11]{angluin2020}.

We can convert IFAs to mod-2-MAs without incurring a blowup:
\begin{proposition}\label{prop-IFA2mod-2}
For any IFA $\A$ with $n$ states there exists a mod-2-MA $\A'$ of at most $n$ states with $L_{\A} =  L_{\A'}$.
\end{proposition}
\begin{proof}
Let $H$ be the Hankel matrix of $L_{\A}$.
By Proposition~\ref{prop-Hankel}, $\rank H \leq n$, where the rank is taken over $\R$.
Invoking Lemma~\ref{lem-R-Q} and Lemma~\ref{lem-Z2-R} then shows that $\rank H \leq n$ also when the rank is taken over $\Z_2$.
Then by Proposition \ref{prop-Hankel} there exists a mod-2-MA with $\rank H \leq n$ (over $\Z_2$) states that accepts the same language as $\A$.
\qed
\end{proof}

However, the converse requires an exponential blowup.
Inspired by Angluin et al.'s~\cite{angluin2020} proof that mod-2 automata can be exponentially more succinct than NFAs, this proof makes use of shift register sequences.
However, note that this proof does not require the assumption that there are infinitely many Mersenne primes.
For further information on shift register sequences, see~\cite{golomb1967}.
A \emph{shift register sequence} of dimension $d$ is an infinite periodic sequence $\{a_n\}$ of bits defined by initial conditions $a_i = b_i$ for $i = 0,\ldots,d-1$ and $b_i \in \{0,1\}$, and a linear recurrence
\begin{equation*}
a_n = c_1a_{n-1}+c_2a_{n-2}+\ldots+c_da_{n-d},
\end{equation*}
for all $n \geq d$, where each $c_i \in \{0,1\}$ and addition is done modulo 2.
The \emph{minimum period} of a periodic sequence $\{a_n\}$ is the lowest $p \in \N$ such that $a_n = a_{n \bmod{p}}$ for every $n$.
The maximum possible minimum period of a shift register sequence is $2^d-1$, and it is known that for each positive integer $d$ there are shift register sequences of maximum period.
These are known as \emph{maximal length} or \emph{pseudo-noise} sequences~\cite{golomb1967}.

Given $d > 0$, let $a_n = c_1a_{n-1}+c_2a_{n-2}+\ldots+c_da_{n-d}$ define a maximum period shift register sequence.
Let $L_d$ be the language over a unary alphabet $\{\#\}$ defined by $\#^n \in L_d$ if and only if $a_n = 1$.
We have the following:
\begin{proposition}\label{prop-mod-22IFA}
The language $L_d$ is accepted by a mod-2-MA with $d$ states, but not by any IFA with fewer than $2^d-1$ states.
\end{proposition}
\begin{proof}
The existence of a mod-2-MA with $d$ states accepting $L_d$ is shown in~\cite[Lemma 11]{angluin2020}.

Consider the Hankel matrix over $\R$ of $L_d$.
Since $\{a_n\}$ is $2^d-1$ periodic, we have that the $i$'th row of the Hankel matrix is equal to the $i+2^d-1$'st row for any $i$, and similar for the columns.
Hence, the rank (over $\R$) of the Hankel matrix is equal to the rank of the submatrix of size $(2^d-1)\times(2^d-1)$ in the top left corner. We will call this matrix $H$.
Notice that $H_{i,j} = a_{i+j}$.
The \emph{auto-correlation} of $\{a_n\}$ is defined as
\begin{equation*}
C(\tau) = \sum_{k=1}^{2^d-1}a_ka_{k+\tau},
\end{equation*}
for $\tau \geq 0$. By Equation 10 on page 82 in~\cite{golomb1967}, we have that $C(\tau) = 2^{d-1}$ if $\tau = 0$ and $C(\tau) = 2^{d-2}$ otherwise.
Consider $H^2$. If $H^2$ has full rank, then so does $H$.
We have that
\begin{eqnarray*}
(H^2)_{i,j} & = & \sum_{k=1}^{2^d-1}H_{i,k}H_{k,j} \\
& = & \sum_{k=1}^{2^d-1}a_{i+k}a_{k+j} \\
& = & \sum_{k=1}^{2^d-1}a_ka_{k+|j-i|} = C(|j-i|),
\end{eqnarray*}
where the indices are taken modulo $2^d-1$.
Hence, $H^2$ is the matrix with $2^{d-1}$ on the diagonal and $2^{d-2}$ elsewhere.
We show that the matrix with $2^{-d+2}-2^{-2d+2}$ on the diagonal and $-2^{-2d+2}$ elsewhere is an inverse of $H^2$.
Let $H' \in \R^{(2^d-1)\times(2^d-1)}$ as follows:
\begin{equation*}
H'_{i,j} = \left\{\begin{array}{ll}
2^{-d+2}-2^{-2d+2} & \textrm{if $i=j$} \\
2^{-2d+2} & \textrm{otherwise}
\end{array}\right.
\end{equation*}
Now we have that $(H^2H')_{i,j} = 1$ if $i = j$:
\begin{eqnarray*}
(H^2H')_{i,i} & = & \sum_{k=1}^{2^d-1} H^2_{i,k}H'_{k,i} \\
& = & 2^{d-1}(2^{-d+2}-2^{-2d+2})+(2^d-2)(2^{d-2}(-2^{-2d+2}) \\
& = & 2^{1}-2^{-d+1}-2^{0}+2^{-d+1} \\
& = & 1
\end{eqnarray*}
We also have $(H^2H')_{i,j} = 0$ if $i \neq j$:
\begin{eqnarray*}
(H^2H')_{i,j} & = & \sum_{k=1}^{2^d-1} H^2_{i,k}H'_{k,j} \\
& = & 2^{d-1}(-2^{-2d+2})+2^{d-2}(2^{-d+2}-2^{-2d+2})+(2^d-3)(2^{d-2}(-2^{-2d+2})) \\
& = & -2^{-d+1}+2^{0}-2^{-d}-2^{0}+3*2^{-d} \\
& = & 0
\end{eqnarray*}
Thus, $H'$ is an inverse of $H^2$.
Hence, $H^2$ and $H$ have full rank and thus the Hankel matrix of $L_d$ has rank $2^d-1$.
This means that the smallest IFA accepting $L_d$ has $2^d-1$ states.
\qed
\end{proof}

\section{Image-binary Büchi Automata}\label{sub-iba}

\subsection{Definitions}\label{sub-iba-defs}

Let $\mathcal{A} = (Q, \Sigma, M, \alpha)$ be like in a weighted automaton over a field $\mathbb{F}$ and let $F$ be a set of final states. We call $\mathcal{A}$ \emph{ultimately stable} if for any $q, q' \in Q$ and $a \in \Sigma$ such that there exists a word $w$ with $M(w)_{q',q} \neq 0$ (i.e., there is a path from $q'$ to $q$ over some word $w$), $M(a)_{q,q'} = 0$ or $M(a)_{q,q'} = 1$, meaning that any edges in a loop have weight 1. For any infinite word $w = w_0w_1\ldots$ we call $q_0q_1\ldots \in Q^\omega$ a \emph{path} over $w$ if $\alpha(q_0) \neq 0$ and for all $i$, $M(w_i)_{q_i,q_{i+1}} \neq 0$. We call $q_0q_1\ldots$ a \emph{final path} if $\mathit{infinite}(q_0q_1\ldots) \cap F \neq \emptyset$, where $\mathit{infinite}(q_0q_1\ldots)$ denotes the set of states in $Q$ that occur infinitely often in the path. We will write $\mathit{FinalPaths}_{\A}(w)$ to denote the set of final paths of an automaton $\A$ over a word $w$.

It is clear that for any path $q_0q_1\ldots$ over a word $w = w_0w_1\ldots$ there exists an $i$ such that for any $j \geq i$, $q_j$ lies on a loop, and therefore we can define the weight of the path $q_0q_1\ldots$ over $w$ to be $\lim_{i \to \infty} \prod_{n \leq i} M(w_n)_{q_n,q_{n+1}}$, denoted by $\mathit{weight}(q_0q_1\ldots, w)$. If $w$ is clear from context, we may simply write $\mathit{weight}(q_0q_1\ldots)$. For any word $w$ with finitely many final paths, we define the weight of $w$ to be the sum of the weights of the final paths over $w$, denoted by $L_\mathcal{A}(w)$. We call $\mathcal{A} = (Q, \Sigma, M, \alpha, F)$ an \emph{image-binary (weighted) Büchi automaton (IBA)} if it is ultimately stable, there exists a bound $N \in \N$ such that $|\mathit{FinalPaths}_{\mathcal{A}}(w)| \leq N$ for any word $w$, and $L_\mathcal{A}(\Sigma^\omega) \subseteq \{0,1\}$, i.e. for all $w \in \Sigma^\omega$, $L_\mathcal{A}(w) \in \{0,1\}$.

If an IBA $\mathcal{A}$ is such that $\alpha \in \{0,1\}^Q$ and $M(a) \in \{0,1\}^{Q \times Q}$ for all $a \in \Sigma$, then $\mathcal{A}$ is called an \emph{unambiguous Büchi automaton (UBA)}. Similarly to the finite word case, we note that this definition of a UBA is essentially equivalent to the classical one, which says that a UBA is an NBA (nondeterministic Büchi automaton) where each word has at most 1 final path. We also see that a \emph{deterministic Büchi automaton (DBA)} essentially is a special case of a UBA, and hence of an IBA.

We will use the following notation: given a finite sequence $a = a_1a_2\ldots a_n$, we will write $\mathord{last}(a)$ to denote $a_n$. Given a (possibly finite) sequence $a = a_1a_2\ldots$ and a character $a_0$ we will write $a_0\cdot a$ to denote the concatenation $a_0a_1a_2\ldots$. We will write $\Bin$ to denote the two element set $\{\bot,\top\}$ and for any set $S$ we will write $\pow^{\leq k}(S)$ to denote the set $\{S' \subseteq S \; \mid \; |S'| \leq k\}$. For ease of notation we will treat $\Bin$ as the set of true ($\top$) and false ($\bot$) predicates, on which the standard Boolean operations apply.

\subsection{IBAs and $k$-Ambiguous NBAs}\label{sub-model-kaba}

In this section we introduce $k$-ambiguous NBAs ($k$-ABAs) and show that they can be exponentially more concise than IBAs. We give a procedure to translate a $k$-ABA into an equivalent IBA using a PSPACE transducer.

A non-deterministic Büchi automaton (NBA) is a tuple $(Q, \Sigma, \delta, Q_0, F)$ where $Q$ is a state set, $\Sigma$ is an alphabet, $\delta : Q \times \Sigma \rightarrow Q$ is a transition relation, $Q_0 \subseteq Q$ is a set of initial states, and $F$ is a set of final states. For any infinite word $w = w_0w_1\ldots$ we call $q_0q_1\ldots \in Q^\omega$ a \emph{path} over $w$ if $q_0 \in Q_0$ and for all $i$, $q_{i+1} \in \delta(q_i, w_i)$. We call $q_0q_1\ldots$ a \emph{final path} if $\mathit{infinite}(q_0q_1\ldots) \cap F \neq \emptyset$, where $\mathit{infinite}(q_0q_1\ldots)$ denotes the set of states in $Q$ that occur infinitely often in the path. We will write $\mathit{FinalPaths}_{\A}(w)$ to denote the set of final paths of an automaton $\A$ over a word $w$. The \emph{language} of an NBA is the set of those words $w$ such that $\mathit{FinalPaths}_{\A}(w) \neq \emptyset$. A $k$-ABA is an NBA such that for every word $w$, $|\mathit{FinalPaths}_{\A}(w)| \leq k$. For the rest of the section, fix a $k$-ambiguous NBA $\kA = (Q, \Sigma, \delta, Q_0, F)$.


We have that $k$-ABAs can be exponentially more succinct than equivalent IBAs:
\begin{lemma}\label{lem:exponlower}
Let $\mathcal{A}$ be a $k$-ABA with $n$ states. The minimal IBA accepting the same language as $\mathcal{A}$ may require at least $2^n$ states, even if $k = n$.
\end{lemma}
\begin{proof}
By Proposition \ref{prop-NFA2IFA}, there exists a family of $n$-ambiguous NFAs with exponentially fewer states than the minimal IFAs accepting the same languages.
Let $\mathcal{A}$ be such an NFA on $n$ states, and let $L \subseteq \Sigma^*$ be its language.
Consider the language $(L)_\$$ given by $w \in (L)_\$$ if and only if $w$ can be decomposed as $u\$v$, where $\$$ does not occur in $\Sigma$ and $u \in L$.
We can create an $n$-ABA that accepts this language with $n+1$ states by adding a $\$$-labeled arrow from the final states of $\mathcal{A}$ to an accepting sink state.
However, let $\mathcal{A}'$ be any IBA accepting $(L)_\$$.
Let $\$w$ be any infinite word starting in $\$$, and let $\eta_q = L_{\mathcal{A'}[q]}(\$w)$.
Let $\mathcal{A''}$ be IFA defined as follows: the state set of $\mathcal{A''}$ is equal to the state set of $\mathcal{A'}$, the alphabet is the alphabet of $\mathcal{A}$ (i.e. the alphabet of $\mathcal{A'}$ without $\$$), there exists an $a$-edge between two states $q, q'$ if and only if there exists an $a$-edge between $q$ and $q'$ in $\mathcal{A'}$, and the final vector is given by $\eta$.
We claim that $\mathcal{A''}$ accepts $L$.
Indeed, we have $L_{\mathcal{A''}}(u) = L_{\mathcal{A'}}(u\$w)$.
Hence, $\mathcal{A'}$ has at least $2^n$ states, because otherwise $\mathcal{A''}$ would be an IFA with fewer than $2^n$ states accepting $L$.
\qed
\end{proof}

The rest of this section will be dedicated to converting $k$-ABAs to equivalent IBAs, resulting in an IBA with at most a singly exponential state set size blowup.

By the binomial theorem, $(1+x)^n = \sum_{i=0}^n {n \choose i}x^i$, and hence, setting $x = -1$, $1 = 1 - \sum_{i=0}^n (-1)^i{n \choose i} = 1 - (-1)^0{n \choose 0} - \sum_{i=1}^n (-1)^i{n \choose i} = \sum_{i=1}^n (-1)^{i-1}{n \choose i}$. Hence, for any set $S$, $\sum_{S' \in \pow(S) \setminus \{\emptyset\}} (-1)^{|S'|-1} = \sum_{i=1}^{|S|} (-1)^{i-1} {|S| \choose i} = 1$. Let $R$ be the set of final paths of $\kA$ over a word $w$. We can design an (infinite state) IBA $\kA' = (Q', \Sigma, \Delta', \alpha, F')$ where final paths correspond to subsets of $R$, and where for each $R' \subseteq R$, the final path corresponding to $R'$ has weight $(-1)^{|R'|-1}$. This IBA is given as follows:
\begin{itemize}
	\item $Q' = \pow^{\leq k}(Q^* \times \Bin) \setminus \{\emptyset\}$,
	\item $\alpha_P = (-1)^{|P|-1}$ for each $P \in \pow(Q_0 \times \{\bot\}) \setminus \{\emptyset\}$ and $\alpha_P = 0$ otherwise,
	\item $F' = \pow(Q^* \times \{\top\})$, and
	\item for any $P \in Q'$, let $b''= \bot$ if for all $(r, b) \in P$, $b = \top$, and let $b'' = \top$ otherwise. Let $P'$ be such that:\
	\begin{itemize}
		\item For any $(r, b) \in P$, there exists $(r', b') \in P'$ and $q \in \delta(\mathord{last}(r), a)$ such that $r' = r \cdot q$ and $b' = ((\mathord{last}(r) \in F) \vee b) \wedge b''$, and
		\item For any $(r', b') \in P'$, there exists $(r, b) \in P$ and $q \in \delta(\mathord{last}(r), a)$ such that $r' = r \cdot q$ and $b' = ((\mathord{last}(r) \in F) \vee b) \wedge b''$.
	\end{itemize}
	Then $\Delta(a)_{P, P'} = (-1)^{|P'|-|P|}$. For any other $P'$, $\Delta(a)_{P, P'} = 0$.
\end{itemize}
Intuitively, the bit $b$ in a $(r,b)$-tuple flips to true every time $r$ reaches a final state, and back to false every time all the prefixes have reached a final state. This ensures that every sequence of prefixes in a final path of $\kA'$ visits final states infinitely often. This technique mirrors for instance Safra's construction~\cite{safra1988}.

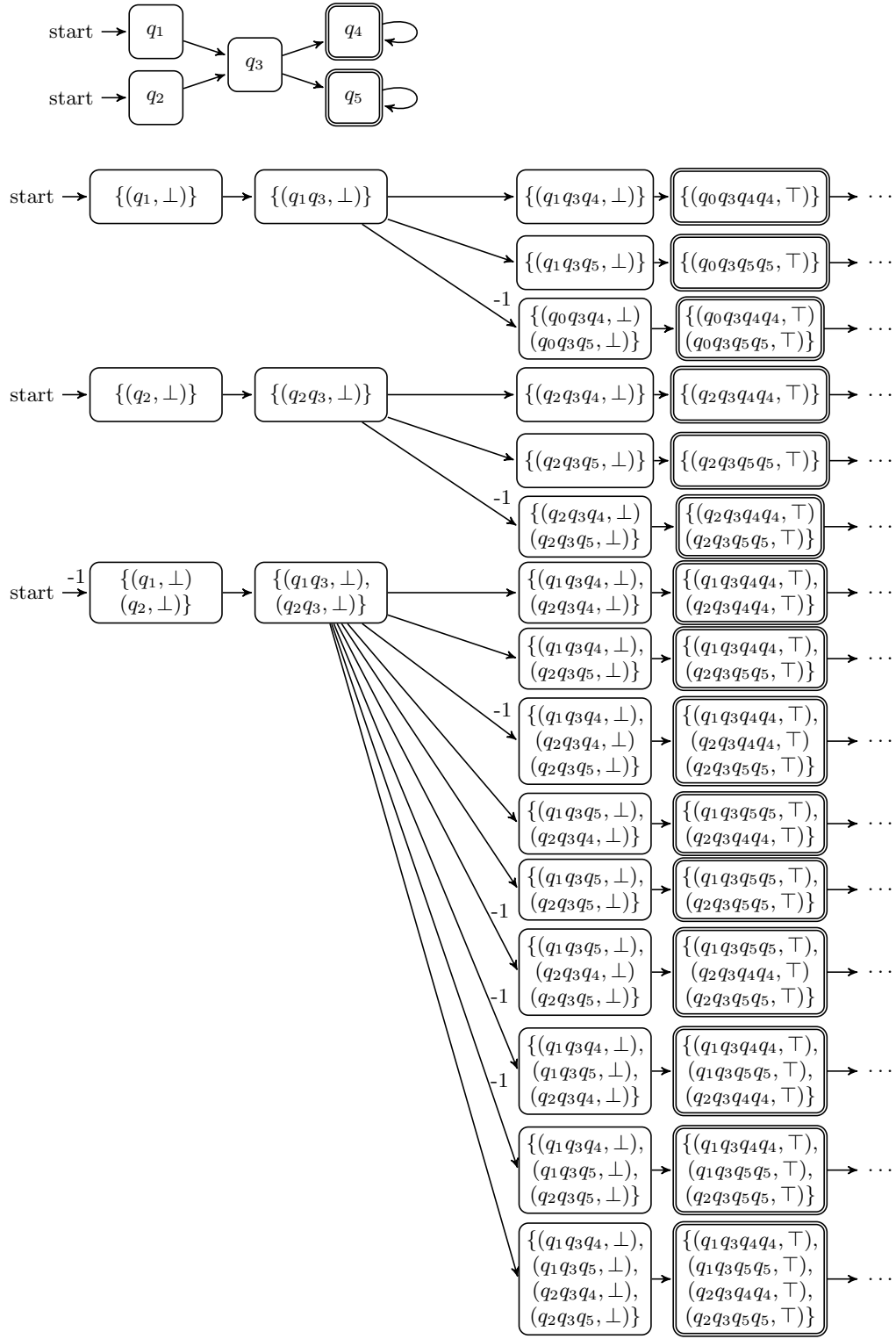
\begin{figure}
\begin{tikzpicture}[->,>=stealth',shorten >= 1pt,auto,node distance=1.5cm,semithick]
\tikzstyle{every state}=[rectangle,rounded corners,draw=black,text=black]
\tikzstyle{pnode}=[minimum width=2cm]

\node[initial,state] (A) at (0,0.5) {$q_1$};
\node[initial,state] (B) at (0,-0.5) {$q_2$};
\node[state] (C) at (1.5,0) {$q_3$};
\node[accepting,state] (D) at (3,0.5) {$q_4$};
\node[accepting,state] (E) at (3,-0.5) {$q_5$};

\path (A) edge (C);
\path (B) edge (C);
\path (C) edge (D);
\path (C) edge (E);
\path (D) edge [loop right] (D);
\path (E) edge [loop right] (E);

\node[initial, state,pnode] (a) at (0, -2) {$\{(q_1, \bot)\}$};
\node[state,pnode] (ac) at (2.5,-2) { $\{(q_1q_3, \bot)\}$};
\node[state,pnode] (acd) at (6.5,-2) { $\{(q_1q_3q_4,\bot)\}$};
\node[state,pnode] (ace) at (6.5,-3) { $\{(q_1q_3q_5,\bot)\}$};
\node[state,pnode,align=center] (acde) at (6.5,-4) { $\{(q_0q_3q_4,\bot)$\\$(q_0q_3q_5,\bot)\}$};
\node[accepting,state,pnode] (acdd) at (9,-2) { $\{(q_0q_3q_4q_4,\top)\}$};
\node[accepting,state,pnode] (acee) at (9,-3) { $\{(q_0q_3q_5q_5,\top)\}$};
\node[accepting,state,pnode,align=center] (acddee) at (9,-4) { $\{(q_0q_3q_4q_4,\top)$\\$(q_0q_3q_5q_5,\top)\}$};

\node[initial, state,pnode] (b) at (0, -5) {$\{(q_2, \bot)\}$};
\node[state,pnode] (bc) at (2.5,-5) { $\{(q_2q_3, \bot)\}$};
\node[state,pnode] (bcd) at (6.5,-5) { $\{(q_2q_3q_4,\bot)\}$};
\node[state,pnode] (bce) at (6.5,-6) { $\{(q_2q_3q_5,\bot)\}$};
\node[state,pnode,align=center] (bcde) at (6.5,-7) { $\{(q_2q_3q_4,\bot)$\\$(q_2q_3q_5,\bot)\}$};
\node[accepting,state,pnode] (bcdd) at (9,-5) { $\{(q_2q_3q_4q_4,\top)\}$};
\node[accepting,state,pnode] (bcee) at (9,-6) { $\{(q_2q_3q_5q_5,\top)\}$};
\node[accepting,state,pnode,align=center] (bcddee) at (9,-7) { $\{(q_2q_3q_4q_4,\top)$\\$(q_2q_3q_5q_5,\top)\}$};

\node[initial, state,pnode,align=center] (ab) at (0, -8) {$\{(q_1, \bot)$\\$(q_2,\bot)\}$};
\node (-1) at (-1.2,-7.75) {-1};
\node[state,pnode,align=center] (abc) at (2.5,-8) { $\{(q_1q_3,\bot),$\\ $(q_2q_3,\bot)\}$};

\node[state,pnode,align=center] (abcdd) at (6.5,-8) { $\{(q_1q_3q_4,\bot),$\\ $(q_2q_3q_4,\bot)\}$};
\node[state,pnode,align=center] (abcde) at (6.5,-9) { $\{(q_1q_3q_4,\bot),$\\ $(q_2q_3q_5,\bot)\}$};
\node[state,pnode,align=center] (abcdde) at (6.5,-10.25) { $\{(q_1q_3q_4,\bot),$\\ $(q_2q_3q_4,\bot)$\\$(q_2q_3q_5,\bot)\}$};

\node[state,pnode,align=center] (abced) at (6.5,-11.5) { $\{(q_1q_3q_5,\bot),$\\ $(q_2q_3q_4,\bot)\}$};
\node[state,pnode,align=center] (abcee) at (6.5,-12.5) { $\{(q_1q_3q_5,\bot),$\\ $(q_2q_3q_5,\bot)\}$};
\node[state,pnode,align=center] (abcede) at (6.5,-13.75) { $\{(q_1q_3q_5,\bot),$\\ $(q_2q_3q_4,\bot)$\\$(q_2q_3q_5,\bot)\}$};

\node[state,pnode,align=center] (abcded) at (6.5,-15.25) { $\{(q_1q_3q_4,\bot),$\\$(q_1q_3q_5,\bot),$\\ $(q_2q_3q_4,\bot)\}$};
\node[state,pnode,align=center] (abcdee) at (6.5,-16.75) { $\{(q_1q_3q_4,\bot),$\\$(q_1q_3q_5,\bot),$\\ $(q_2q_3q_5,\bot)\}$};
\node[state,pnode,align=center] (abcdede) at (6.5,-18.4) { $\{(q_1q_3q_4,\bot),$\\$(q_1q_3q_5,\bot),$\\ $(q_2q_3q_4,\bot),$\\$(q_2q_3q_5,\bot)\}$};

\node[accepting,state,pnode,align=center] (abcddp) at (9,-8) { $\{(q_1q_3q_4q_4,\top),$\\ $(q_2q_3q_4q_4,\top)\}$};
\node[accepting,state,pnode,align=center] (abcdep) at (9,-9) { $\{(q_1q_3q_4q_4,\top),$\\ $(q_2q_3q_5q_5,\top)\}$};
\node[accepting,state,pnode,align=center] (abcddep) at (9,-10.25) { $\{(q_1q_3q_4q_4,\top),$\\ $(q_2q_3q_4q_4,\top)$\\$(q_2q_3q_5q_5,\top)\}$};

\node[accepting,state,pnode,align=center] (abcedp) at (9,-11.5) { $\{(q_1q_3q_5q_5,\top),$\\ $(q_2q_3q_4q_4,\top)\}$};
\node[accepting,state,pnode,align=center] (abceep) at (9,-12.5) { $\{(q_1q_3q_5q_5,\top),$\\ $(q_2q_3q_5q_5,\top)\}$};
\node[accepting,state,pnode,align=center] (abcedep) at (9,-13.75) { $\{(q_1q_3q_5q_5,\top),$\\ $(q_2q_3q_4q_4,\top)$\\$(q_2q_3q_5q_5,\top)\}$};

\node[accepting,state,pnode,align=center] (abcdedp) at (9,-15.25) { $\{(q_1q_3q_4q_4,\top),$\\$(q_1q_3q_5q_5,\top),$\\ $(q_2q_3q_4q_4,\top)\}$};
\node[accepting,state,pnode,align=center] (abcdeep) at (9,-16.75) { $\{(q_1q_3q_4q_4,\top),$\\$(q_1q_3q_5q_5,\top),$\\ $(q_2q_3q_5q_5,\top)\}$};
\node[accepting,state,pnode,align=center] (abcdedep) at (9,-18.4) { $\{(q_1q_3q_4q_4,\top),$\\$(q_1q_3q_5q_5,\top),$\\ $(q_2q_3q_4q_4,\top),$\\$(q_2q_3q_5q_5,\top)\}$};

\node (d1) at (11,-2) {\ldots};
\node (d2) at (11,-3) {\ldots};
\node (d3) at (11,-4) {\ldots};
\node (d4) at (11,-5) {\ldots};
\node (d5) at (11,-6) {\ldots};
\node (d6) at (11,-7) {\ldots};
\node (d7) at (11,-8) {\ldots};
\node (d8) at (11,-9) {\ldots};
\node (d9) at (11,-10.25) {\ldots};
\node (d10) at (11,-11.5) {\ldots};
\node (d11) at (11,-12.5) {\ldots};
\node (d12) at (11,-13.75) {\ldots};
\node (d13) at (11,-15.25) {\ldots};
\node (d14) at (11,-16.75) {\ldots};
\node (d15) at (11,-18.4) {\ldots};

\path (a) edge (ac);
\path (ac) edge (acd.west);
\path (ac) edge (ace.west);
\path (ac) edge node[pos=0.9,above=2] {-1} (acde.west);
\path (acd) edge (acdd);
\path (ace) edge (acee);
\path (acde) edge (acddee);

\path (b) edge (bc);
\path (bc) edge (bcd.west);
\path (bc) edge (bce.west);
\path (bc) edge node[pos=0.9,above=2] {-1} (bcde.west);
\path (bcd) edge (bcdd);
\path (bce) edge (bcee);
\path (bcde) edge (bcddee);

\path (ab) edge (abc);
\path (abc) edge (abcdd.west);
\path (abc) edge (abcde.west);
\path (abc) edge node[pos=0.9,above=2] {-1} (abcdde.west);
\path (abc) edge (abced.west);
\path (abc) edge (abcee.west);
\path (abc) edge node[pos=0.9,above=5] {-1} (abcede.west);
\path (abc) edge node[pos=0.9,above=7] {-1} (abcded.west);
\path (abc) edge node[pos=0.9,above=9] {-1} (abcdee.west);
\path (abc) edge (abcdede.west);
\path (abcdd) edge (abcddp);
\path (abcde) edge (abcdep);
\path (abcdde) edge (abcddep);
\path (abced) edge (abcedp);
\path (abcee) edge (abceep);
\path (abcede) edge (abcedep);
\path (abcded) edge (abcdedp);
\path (abcdee) edge (abcdeep);
\path (abcdede) edge (abcdedep);

\path (acdd) edge (d1);
\path (acee) edge (d2);
\path (acddee) edge (d3);
\path (bcdd) edge (d4);
\path (bcee) edge (d5);
\path (bcddee) edge (d6);
\path (abcddp) edge (d7);
\path (abcdep) edge (d8);
\path (abcddep) edge (d9);
\path (abcedp) edge (d10);
\path (abceep) edge (d11);
\path (abcedep) edge (d12);
\path (abcdedp) edge (d13);
\path (abcdeep) edge (d14);
\path (abcdedep) edge (d15);

\end{tikzpicture}
\caption{An example $4$-ambiguous automaton (above) and its infinite IBA counterpart (below). In the IBA, the weights of the unlabeled edges are 1.}\label{fig:kaba}
\end{figure}

In~\Cref{fig:kaba} we give an example of a 4-ambiguous automaton over a unary alphabet, together with the corresponding infinite $\kA'$.

\begin{lemma}\label{lem:kapequiv}
$\kA'$ is an infinite-state IBA equivalent to $\kA$.
\end{lemma}
\begin{proof}
For $\kA'$ to be an IBA equivalent to $\kA$, it needs to satisfy three properties - the edges of $\kA'$ over loops have weight 1, any word $w$ has at most $N$ final paths for some global bound $N$, and the sum of the weights of final paths over $w$ is 1 if $w$ is accepted by $\kA$ and 0 otherwise. Since states in a path of $\kA'$ consist of sets of prefixes of increasing length, we see that $\kA'$ does not have any loops and hence the first property is trivially true. For the second and third properties, let $w$ be any fixed word.

Suppose $w$ is not accepted by $\kA$, then by König's lemma there exists a bound $m$ such that any path of $\kA$ over $w$ has at most $m$ occurences of final states. Let $\rho_1\rho_2\ldots\rho_i$ be any prefix of a path of $\kA'$ over $w$. Since for any $(r,b) \in \rho_i$, $r$ has at most $m$ occurrences of final states, and all the bits are flipped to false any time every prefix in $\rho_1\rho_2\ldots\rho_i$ has visited a final state, we see that $\rho_1\rho_2\ldots\rho_i$ also visits final states at most $m$ times. Therefore, $\kA'$ has no final paths (and hence the sum of the weights is 0) and the second and third properties are true.

Suppose, then, that $w$ is accepted by $\kA$. Let $\rho_1\rho_2\ldots$ be a path of $\kA'$ over $w$. Suppose that for some $i$ there exists $(r,b) \in \rho_i$ such that $r$ is not a prefix of a final path of $\kA$ over $w$. Let $r$ be the shortest such prefix.
Then for large enough $m$ and any $m' > m$ there exists $(r', \bot) \in \rho_{m'}$ where $r$ is a prefix of $r'$. Hence, $\rho_{m'} \not \in F'$ and $\rho_1\rho_2\ldots$ is not a final path.

Finally, let $\rho_1\rho_2\ldots$ be a path of $\kA'$ over $w$ such that for every $i$ and every $(r,b) \in \rho_i$, $r$ is a prefix of a final path of $\kA$ over $w$. Let $R$ be the set of those paths $\rho'$ of $\kA$ over $w$ such that for all $i$ there exists a tuple $(r,b) \in \rho_i$ where $r$ is a prefix of $\rho'$. We have that for any $i$ and any $(r,b) \in \rho_i$ there exists $(r', b') \in \rho_{i+1}$ such that $r' = r \cdot q$ with $q \in \delta(\mathord{last}(r), w_i)$ and for any $(r', b') \in \rho_{i+1}$ there exists $(r, b) \in \rho_i$ such that $r' = r \cdot q$ with $q \in \delta(\mathord{last}(r), w_i)$. Therefore $R$ is nonempty and $\rho_i$ is precisely the set of prefixes of length $i$ of paths in $R$ (paired with some boolean $b$). Hence, paths $\rho_1\rho_2\ldots$ in $\kA'$ over $w$ correspond uniquely to sets of paths of $\kA$ over $w$.

We claim that $R$ is a set of final paths, and that the weight of $\rho_1\rho_2\ldots$ is $(-1)^{|R|-1}$. Indeed, suppose $\rho'_1\rho'_2\ldots \in R$ is not a final path. Since for any $i$, $\rho'_1\rho'_2\ldots\rho'_i$ prefixes a final path, there exists $j > i$ such that $\rho'_1\rho'_2\ldots\rho'_i$ is a prefix of a final path $\rho''_1\rho''_2\ldots$, and $\rho''_j \neq \rho'_j$. Hence $\kA$ has infinitely many final paths over $w$, which contradicts its $k$-ambiguity. Moreover, the weight of any prefix of length $i$ of $\rho_1\rho_2\ldots$ is $(-1)^{|\rho_1|-1}\prod_j (-1)^{|\rho_{j+1}|-|\rho_j|} = (-1)^{|\rho_i|-1}$. Since the size of $\rho_i$ is non-decreasing and bounded from above by $k$, and for any $i$, $\rho_i$ contains prefixes of length $i$ of paths in $R$, for large enough $i$ we have that $|\rho_i| = |R|$, and hence the weight of $\rho_1\rho_2\ldots$ is $(-1)^{|R|-1}$.

Thus, final paths in $\kA'$ over $w$ correspond uniquely to sets of final paths in $\kA$ over $w$, and hence there are at most $2^k$ final paths in $\kA'$ over $w$. Moreover, since by the binomial theorem we have $\sum_{i=1}^{|S|} (-1)^{i-1} {|S| \choose i} = 1$ for any set $S$, and any path in $\kA'$ corresponding to a set $R$ of final paths in $\kA$ has weight $(-1)^{|R|-1}$, we see that the sum of the weights of final paths of $\kA'$ over $w$ is precisely equal to 1 if and only if $\kA$ has an accepting path over $w$.
\qed
\end{proof}

We will construct a finite IBA called the $k$-disambiguation of $\kA$ based on $\kA'$ that accepts the same language as $\kA$.

Let $\pi_{\mathord{last}} : Q' \rightarrow [k]^{Q \times \Bin}$ be defined as $\pi_{\mathord{last}}(P)_{q,b} = |\{(r, b) \in P \mid \mathord{last}(r) = q\}|$.
We extend $\pi_{\mathord{last}}$ over finite and infinite sequences of elements of $Q'$ in the natural way: $\pi_{\mathord{last}}(P_1P_2\ldots) = \pi_{\mathord{last}}(P_1)\cdot\pi_{\mathord{last}}(P_2\ldots)$.
\begin{lemma}\label{lem:finalvecs}
Let $\rho=\rho_1\rho_2\ldots \in (Q')^\omega$ be a path of $\kA'$ over a word $w$ and let $\rho'_1\rho'_2\ldots = \pi_{\mathord{last}}(\rho)$. Then $\rho_1\rho_2\ldots$ is final if and only if for infinitely $i$, $(\rho'_i)_{q,\bot} = 0$ for all $q \in Q$.
\end{lemma}
\begin{proof}
Trivial.
\qed
\end{proof}
 Paths in our $k$-disambiguation will be sequences in $([k]^{Q \times \Bin})^\omega$ such that there exist paths in $\kA'$ that map to that sequence. However, there is not a one-to-one correspondence between sequences over $[k]^{Q \times \Bin}$ and paths in $\kA'$: in Figure \ref{fig:kaba}, for instance, both
$\{(q_1,\bot),(q_2,\bot)\}$ $\{(q_1q_3,\bot),(q_2q_3,\bot)\}$ $\{(q_1q_3q_4,\bot),(q_2q_3q_5,\bot)\}$ $\{(q_1q_3q_4q_4,\top),(q_2q_3q_5q_5,\top)\}\ldots$ and  $\{(q_1,\bot),(q_2,\bot)\}$ $\{(q_1q_3,\bot),(q_2q_3,\bot)\}$ $\{(q_1q_3q_5,\bot),(q_2q_3q_4,\bot)\}$ $\{(q_1q_3q_5q_5,\top),(q_2q_3q_4q_4,\top)\}\ldots$ map to the sequence $(1,1,0,0,0,0,0,0,0,0)^\top$ $(0,0,2,0,0,0,0,0,0,0)^\top$ $((0,0,0,1,1,0,0,0,0,0)^\top$ $(0,0,0,0,0,0,0,0,1,1)^\top)^\omega$, where the order of the vector components is given by $((q_1,\bot),(q_2,\bot),(q_3,\bot),(q_4,\bot),(q_5,\bot),(q_1,\top),(q_2,\top),(q_3,\top),(q_4,\top),(q_5,\top))$.

Fix any $\vec{r}, \vec{r'} \in [k]^{Q\times \Bin}$ and $a \in \Sigma$. Let $P$ be any state in $Q'$ such that $\pi_{\mathord{last}}(P) = \vec{r}$. As it turns out, the number of states $P'$ with $\pi_{\mathord{last}}(P') = \vec{r'}$ where $P'$ is an $a$-successor of $P$ only depends on $\vec{r}$, $a$, and $\vec{r'}$. We call this number $w(\vec{r},a,\vec{r'})$.
\begin{lemma}\label{lem:succnum}
The number $w(\vec{r},a,\vec{r'})$ is unique and at most exponential in $k$.
\end{lemma}
\begin{proof}
Given $\vec{r}, \vec{r'} \in [k]^{Q \times \B}$ and $a \in \Sigma$, let
\begin{eqnarray*}
C = \{f : Q \times \B \times [k] \rightarrow \pow(Q) \mid & \forall (q,b,i) \in Q \times \B\times[k] : \\
& ((i > \vec{r}_{q,b}) \rightarrow f(q,b,i) = \emptyset) \land \\
&  ((i \leq \vec{r}_{q,b}) \rightarrow f(q,b,i) \in \pow(\delta(q,a)) \setminus \{\emptyset\})\}.
\end{eqnarray*}
Intuitively, functions in $C$ map the final states of prefixes in a state $P$ with $\pi_{\mathord{last}}(P) = \vec{r}$ to nonempty sets of successor states. Let $b'' = \bot$ if $b = \top$ for all $(r,b) \in P$, and let $b'' = \top$ otherwise. Let $C'$ be the set of those $f \in C$ such that for any $(q',b') \in Q \times \B$, $|\{(q,b,i) \in Q \times \B \times [k] \mid (b' = (q \in F \lor b) \land b'') \land q' \in f(q,b,i)\}| = \vec{r'}_{q',b'}$. We claim that $w(\vec{r},a,\vec{r'}) = |C'|$.

Fix an ordering on $Q^*$. Since $\pi_{\mathord{last}}(P) = \vec{r}$, we see that for any $q \in Q, b \in \B$ there are $\vec{r}_{q,b}$ tuples $(s,b)$ in $P$ where $\mathord{last}(s) = q$. Let $s_{q,b,i}$ denote the $i$'th (under the ordering on $Q^*$) path in $P$ that is paired with bit $b$ and ends in state $q$. Then for any $f \in C$, $f(q,b,i)$ is a nonempty set of $a$-successors of $q$. Let $b'' = \bot$ if $b' = \top$ for every $(s', b') \in P$ and let $b'' = \bot$ otherwise. If $q \in F$, let $b' = \top \land b''$, and otherwise let $b' = b \land b''$. Hence, for any $f \in C$ the following is an $a$-successor of $P$ in $\kA'$:
\begin{equation*}
P' = \left\{(s_{q,b,i}\cdot q', b') \mid (s_{q,b,i},b)\in P \land q' \in f(q,b,i)\right\}.
\end{equation*}
If $f \in C'$, then by definition of $C'$, for any $q', b'$, $\pi_{\mathord{last}}(P')_{q',b'} = \vec{r'}_{q',b'}$, and hence $\pi_{\mathord{last}}(P') = \vec{r'}$. Since each $f \in C$ gives rise to a unique successor, this means $P$ has $|C'|$ $a$-successors $P'$ with $\pi_{\mathord{last}}(P') = \vec{r'}$. Moreover, $|C'| \leq |C| \leq 2^{2k|Q|^2}$.
\qed
\end{proof}

For a state $\vec{r} \in Q''$ we will write $\mathord{size}(\vec{r}) = \sum_{(q,b)} \vec{r}_{q,b}$ to denote the size of $\vec{r}$. We define the IBA $\kdis'(\kA) = (Q'', \Sigma, \Delta'', \alpha', F'')$ where:
\begin{itemize}
	\item $Q'' = ([k]^{Q \times \Bin}) \setminus \{\vec{0}\}$,
	\item $\Delta''(a)_{\vec{r},\vec{r'}} = (-1)^{\mathord{size}(\vec{r'})-\mathord{size}(\vec{r})}w(\vec{r},a,\vec{r'})$
	\item $\alpha'_{\vec{r}} = (-1)^{\mathord{size}(\vec{r}) - 1}$ for every $\vec{r}$ with:
	\begin{itemize}
		\item $\vec{r}_{q,b} = 0$ if $q \not \in Q_0$ or $b = \top$, and
		\item $\vec{r}_{q, b} \leq 1$ otherwise.
	\end{itemize}
	$\alpha'_{\vec{r}} = 0$ otherwise.
	\item $F'' = \{\vec{r} \in Q'' \mid \forall (q, \bot) \in Q'' : \vec{r}_{q, \bot} = 0\}$.
\end{itemize}
\begin{figure}
\begin{tikzpicture}[->,>=stealth',shorten >= 1pt,auto,node distance=1.5cm,semithick]
\tikzstyle{every state}=[rectangle,rounded corners,draw=black,text=black]

\node[initial,state,align=center] (10000) at (0,0) {$\{(q_1,\bot)\}$};
\node[initial,state,align=center] (01000) at (0,-3) {$\{(q_2,\bot)\}$};
\node[state,align=center] (00100) at (1,-1.5) {$\{q_3,\bot)\}$};
\node[state,align=center] (00010) at (2,0) {$\{(q_4,\bot)\}$};
\node[accepting,state,align=center] (0001b0) at (4,0) {$\{(q_4,\top)\}$};
\node[state,align=center] (00001) at (3,-1.5) {$\{(q_5,\bot)\}$};
\node[accepting,state,align=center] (00001b) at (5, -1.5) {$\{(q_5,\top)\}$};
\node[state,align=center] (00011) at (2.5,-3) {$\{(q_4,\bot),(q_5,\bot)\}$};
\node[accepting,state,align=center] (0001b1b) at (5.5,-3) {$\{(q_4,\top),(q_5,\top)\}$};

\path (10000) edge (00100);
\path (01000) edge (00100);
\path (00100) edge (00010);
\path (00100) edge (00001);
\path (00010) edge [bend left] (0001b0);
\path (0001b0) edge [bend left] (00010);
\path (00001) edge [bend left] (00001b);
\path (00001b) edge [bend left] (00001);
\path (00100) edge node[above] {-1} (00011);
\path (00011) edge [bend left] (0001b1b);
\path (0001b1b) edge [bend left] (00011);

\node[initial,state,align=center] (11000) at (0,-7.5) {$\{(q_1,\bot),(q_2,\bot)\}$};
\node (-1) at (-1.5,-7) {-1};
\node[state,align=center] (00200) at (3,-7.5) {$\{(q_3,\bot),(q_3,\bot)\}$};
\node[state,align=center] (00002) at (3,-6) {$\{q_5,\bot),(q_5,\bot)\}$};
\node[accepting,state,align=center] (00002b) at (3,-4.5) {$\{(q_5,\top),(q_5,\top)\}$};
\node[state,align=center] (00011) at (6,-6) {$\{(q_4,\bot),(q_5,\bot)\}$};
\node[accepting,state,align=center] (0001b1b) at (6,-4.5) {$\{(q_4,\top),(q_5,\top)\}$};
\node[state,align=center] (00012) at (6,-7.5) {$\{(q_4,\bot),(q_5,\bot)$\\$(q_5,\bot)\}$};
\node[accepting,state,align=center] (0001b2b) at (9,-7.5) {$\{(q_4,\top),(q_5,\top)$\\$(q_5,\top)\}$};
\node[state,align=center] (00020) at (6,-9) {$\{(q_4,\bot),(q_4,\bot)\}$};
\node[accepting,state,align=center] (0002b0) at (6,-10.5) {$\{(q_4,\top),(q_4,\top)\}$};
\node[state,align=center] (00021) at (3,-9) {$\{(q_4,\bot),(q_4,\bot)$\\$(q_5,\bot)\}$};
\node[accepting,state,align=center] (0002b1b) at (3,-10.5) {$\{(q_4,\top),(q_4,\top)$\\$(q_5,\top)\}$};
\node[state,align=center] (00022) at (0,-9) {$\{(q_4,\bot),(q_4,\bot)$\\$(q_5,\bot),(q_5,\bot)\}$};
\node[accepting,state,align=center] (0002b2b) at (0,-10.5) {$\{(q_4,\top),(q_4,\top)$\\$(q_5,\top),(q_5,\top)\}$};

\path (11000) edge (00200);
\path (00200) edge (00002);
\path (00200) edge node[above] {2} (00011);
\path (00200) edge node[above] {-2} (00012);
\path (00200) edge (00020);
\path (00200) edge node[right] {-2} (00021);
\path (00200) edge (00022);

\path (00002) edge[bend left] (00002b);
\path (00002b) edge[bend left] (00002);
\path (00011) edge[bend left] (0001b1b);
\path (0001b1b) edge[bend left] (00011);
\path (00012) edge[bend left] (0001b2b);
\path (0001b2b) edge[bend left] (00012);
\path (00020) edge[bend left] (0002b0);
\path (0002b0) edge[bend left] (00020);
\path (00021) edge[bend left] (0002b1b);
\path (0002b1b) edge[bend left] (00021);
\path (00022) edge[bend left] (0002b2b);
\path (0002b2b) edge[bend left] (00022);

\end{tikzpicture}
\caption{The $k$-disambiguation of the automaton in Figure \ref{fig:kaba}. The elements of $[k]^{Q \times \Bin}$ are represented by multisets. Unlabeled edges have weight 1.}\label{fig:kdiskaba}
\end{figure}
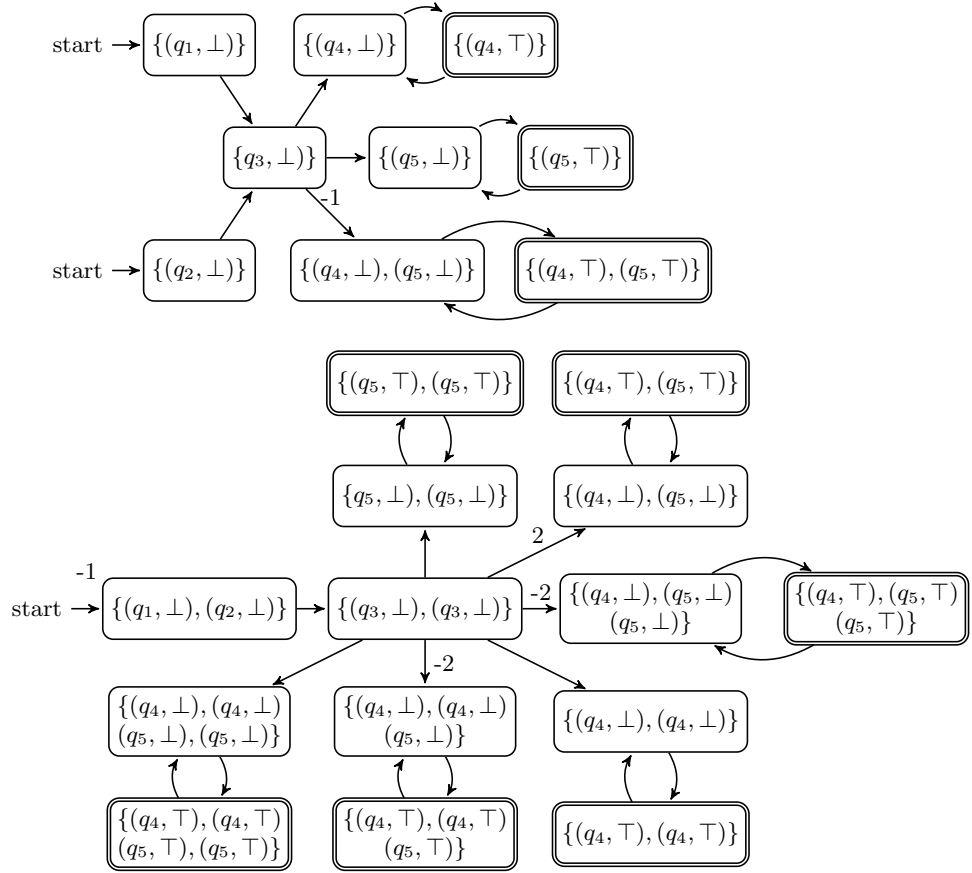 
The IBA $\kdis(\kA)$ (the $k$-disambiguation of $\kA$) is then defined as $\kdis'(\kA)$ restricted to those reachable states in $Q''$ that can reach a loop over a final state. This trimness condition helps the proofs later on, but could be omitted. In~\Cref{fig:kdiskaba} we see the $k$-disambiguation of the automaton in~\Cref{fig:kaba}.

The weights in $\kdis(\kA)$ count the number of equivalent paths in $\kA'$, where two paths $\rho, \rho' \in (Q')^\omega$ are equivalent if $\pi_{\mathord{last}}(\rho) = \pi_{\mathord{last}}(\rho')$:
\begin{lemma}\label{lem:countruns}
Let $\rho \in (Q'')^\omega$ be a path in $\kdis(\kA)$ over a word $w$. Let $R$ be the set of those paths $\rho'$ in $\kA'$ over $w$ such that $\pi_{\mathord{last}}(\rho') = \rho$, and let $n = \max_i \mathord{size}(\rho_i)$. Then $\mathord{weight}(\rho) = (-1)^{n-1}|R|$.
\end{lemma}
\begin{proof}
Let $R_i$ be the set of prefixes of length $i$ of paths in $R$. Let $\rho'_i$ be any path in $R_i$. By Lemma \ref{lem:succnum}, $\rho'_i$ has $w(\rho_i, w_i, \rho_{i+1})$ $w_i$-successors that map to $\rho_{i+1}$ under $\pi_{\mathord{last}}$. Hence, $|R_{i+1}| = w(\rho_i,w_i,\rho_{i+1})|R_i|$ and in particular, $|R| = \lim_{n \to \infty} \prod_{i=1}^n w(\rho_i,w_i,\rho_{i+1})$. Moreover, $\lim_{n \to \infty} (-1)^{\mathord{size}(\rho_1)-1}\prod_{i=1}^n (-1)^{\mathord{size}(\rho_{i+1})-\mathord{size}(\rho_i)} = (-1)^{n-1}$, and hence $\mathord{weight}(\rho) =\lim_{n \to \infty} (-1)^{\mathord{size}(\rho_1)-1}\prod_{i=1}^n (-1)^{\mathord{size}(\rho_{i+1})-\mathord{size}(\rho_i)}w(\rho_i,w_i,\rho_{i+1}) = (-1)^{n-1}|R|$.
\qed
\end{proof}

Since by Lemma \ref{lem:finalvecs}, a path $\rho$ in $\kdis(\kA)$ is final whenever any path $\rho'$ in $\kA'$ with $\pi_{\mathord{last}}(\rho') = \rho$ is, this means $L_{\kdis(\kA)}=L_{\kA'}$. This proves the final result:
\begin{theorem}\label{thm:kaequiv}
$\kdis(\kA)$ is an IBA that accepts the same language as $\kA$.
\end{theorem}
\begin{proof}
We will prove equivalence between $\kA'$ and $\kdis(\kA)$. Combined with Lemma \ref{lem:kapequiv}, this gives us the required result. As before, we need to show that the edges of $\kdis(\kA)$ over loops have weight 1, any word $w$ has at most $N$ final paths for some global bound $N$, and the sum of the weights of final paths over $w$ is 1 if $w$ is accepted by $\kA'$ and 0 otherwise.

Pick any loop in $\kdis(\kA)$ and any $a$-edge from any $\vec{r}$ to any $\vec{r'}$. Note that $w(\vec{r},a,\vec{r'}) > 0$ implies that $\mathord{size}(\vec{r}) \leq \mathord{size}(\vec{r'})$ and hence for any $\vec{r}, \vec{r'}$ on that loop, $\mathord{size}(\vec{r}) = \mathord{size}(\vec{r'})$. Let $w_1\ldots w_m$ be such that there exists a path from an initial state to $\vec{r}$, and let $w_{m+1}\ldots w_n$ be a word that traverses the loop starting with this edge from $\vec{r}$ to $\vec{r'}$ (hence, $w_{m+1} = a$). Finally let $w_{n+1}w_{n+2}\ldots$ be a word such that $w_1\ldots w_m(w_{m+1}\ldots w_{n})^Cw_{n+1}\ldots$ has a final path that traverses the loop over $\vec{r}$ at least $C$ times, such a word exists since every state can reach a loop over a final state. Note that weights in $\kdis(\kA)$ are integers. Then the absolute weight of this path is at least $w(\vec{r},w_{m+1},\vec{r'})^C$. Using Lemmas \ref{lem:countruns} and \ref{lem:finalvecs}, then, we see that there exist at least $w(\vec{r},w_{m+1},\vec{r'})^C$ final paths over $w_1\ldots w_m(w_{m+1}\ldots w_n)^Cw_{n+1}\ldots$. Since this is the case for any $C$, and the number of final paths over any word in $\kA'$ is bounded above, we must have that $w(\vec{r}, w_{m+1}, \vec{r'}) = 1$ and hence $\Delta(a)_{\vec{r},\vec{r'}} = 1$.

Now let $w$ be any word. By Lemma \ref{lem:finalvecs} a path in $\kA'$ is final if and only if the path in $\kdis(\kA)$ it maps to under $\pi_{\mathord{last}}$ is final, and by Lemma \ref{lem:countruns} any path with nonzero weight (i.e. any path) has at least one path in $\kA'$ mapping to it. Hence we see that $\kdis(\kA)$ has at most as many final paths over $w$ as $\kA'$ does, and since the number of final paths in $\kA'$ is globally bounded, so is the number of final paths in $\kdis(\kA)$.

Let $w$ be any word and define the equivalence relation $=_\pi \subseteq Q'^\omega\times Q'^\omega$ as $\rho =_\pi \rho'$ iff $\pi_{\mathord{last}}(\rho) = \pi_{\mathord{last}}(\rho')$ (note that $\pi_{\mathord{last}}(\rho) \in ([k]^{Q\times\Bin})^\omega$). Note that this equivalence relation partitions the final paths of $\kA'$ over $w$, and for any $\rho, \rho'$ with $\rho =_\pi \rho'$, $\mathord{weight}(\rho) = \mathord{weight}(\rho') = (-1)^{\lim_{i \to \infty} |\rho_i|}$. Each equivalence class can be represented by a path in $\kdis(\kA)$. Let $\rho'$ be the representative of the equivalence class containing $\rho$, and let $R$ be the size of this equivalence class. By Lemma \ref{lem:countruns}, $\mathord{weight}(\rho') = \mathord{weight}(\rho)R$. Hence, the sum of the weights of final paths in $\kdis(\kA)$ is equal to the sum of the weights of all the equivalence classes of $=_\pi$ where the representative is a final path, which by Lemma \ref{lem:finalvecs} is precisely the sum of the weights of final paths in $\kA'$.

Hence, the weight of a word $w$ in $\kdis(\kA)$ is equal to the weight of $w$ in $\kA'$. This means $\kdis(\kA)$ is equivalent to $\kA'$ which is in turn equivalent to $\kA$, concluding the proof.
\qed
\end{proof}

\begin{theorem}\label{thm:pspacetrans}
Given a $k$-ambiguous automaton $\kA$ with $n$ states, the disambiguation $\kdis(\kA)$ has at most $k^{2n}$ states. Moreover, $\kdis(\kA)$ can be calculated using a PSPACE transducer.
\end{theorem}
\begin{proof}
The number of states follows from the size of $[k]^{Q \times \B}$.  Note that $k$ is at most singly exponential in the size of the state set of the automaton, which follows from for instance~\cite[Theorem 2.1]{weber1991}, or the fact that finite ambiguity implies the nonexistence of diamonds on a loop. This means that $k^{2n}$ is also singly exponential in $n$. Hence, we can iterate over every element in $Q''$, each of which is a vector in $[k]^{Q \times \B}$ which can be represented in polynomial space. Given $\vec{r}, \vec{r'} \in Q''$ and $a \in \Sigma$, to calculate $\Delta''(a)_{\vec{r},\vec{r'}}$ we need to calculate $(-1)^{\mathord{size}(\vec{r'})-\mathord{size}(\vec{r})}$ and $w(\vec{r},a,\vec{r'})$. Clearly the former can be calculated in polynomial space. For the latter, notice that $w(\vec{r},a,\vec{r'})$ can be doubly exponential, requiring an exponential representation. Hence, we cannot simply list the edges between states in $\kA'$ to calculate $w(\vec{r},a,\vec{r'})$. However, we have the following lemma:
\begin{lemma}\label{lem:wrarcombi}
Let $w(\vec{r},a,\vec{r'})$ as defined above, let $b'' = \bot$ if $\vec{r}_{q,\bot} = 0$ for every $q \in Q$ and $b'' = \top$ otherwise.
For each $(q,b) \in Q\times \B$ let $C_{q,b}$ denote the set of those $\vec{r'}[q,b] \in [k]^{Q \times \B}$ such that
\begin{multline*}
\mathord{size}(\vec{r'}[q,b]) \geq \vec{r}_{q,b}\quad \land \quad
\forall (q',b') \in Q \times \B \ . \ \vec{r'}[q,b]_{q',b'} \leq \vec{r}_{q,b} \quad \land \\
 \forall (q',b') \in Q \times \B \ . \ ((b' \neq (q \in F \lor b) \land b'')\lor q'\not\in\delta(q,a))\rightarrow \vec{r'}[q,b]_{q',b'} = 0\,.
\end{multline*}
Let $\#\mathord{succ}(q,n,\vec{r'}[q,b]) = \sum_{j=0}^n (-1)^j{n \choose j}\left(\prod_{q',b'} {n-j \choose \vec{r'}[q,b]_{q',b'}}\right)$. Then
\begin{equation*}
w(\vec{r},a,\vec{r'}) = \sum_{f : Q \times \B \rightarrow C_{q,b} \mid \sum_{q,b} f(q,b) = \vec{r'}} \prod_{q,b} \#\mathord{succ}(q,\vec{r}_{q,b}, f(q,b)).
\end{equation*}
\end{lemma}
\begin{proof}
We have that $w(\vec{r},a,\vec{r'}) = |C'|$, where $C'$ is defined as in Lemma~\ref{lem:succnum}.
We define the following equivalence relation:
two elements $f, f' \in C'$ are equivalent if for every $(q,b) \in Q\times\B$ we have that $\sum_i [[f(q,b,i)]] = \sum_i [[f'(q,b,i)]]$, where for any set $S$, $[[S]]$ denotes the characteristic vector of $S$.

 Let $\mathcal{G} := \{g : Q \times \B \rightarrow C_{q,b} \; \mid \; \sum_{q,b} g(q,b) = \vec{r'}\}$.
Clearly the equivalence classes of $C'$ can be characterised by elements of $\mathcal{G}$.
We can view any function $f \in C'$ as a way of distributing non-empty successor sets over the numbered elements in $\vec{r}$ in such a way that it adds up to $\vec{r'}$.
Hence, for any $g \in \mathcal{G}$, the number of elements in the equivalence class of $g$ is given by those functions that, for every $(q,b) \in Q \times \B$, distribute the elements in $g(q,b)$ to the $\vec{r}_{q,b}$ paths ending in $(q,b)$.
These distributions are independent, and hence the number of elements in the equivalence class of $g$ is the product over all $(q,b)$ of the number of ways to distribute the elements in $g(q,b)$.

This number is captured by $\#\mathord{succ}(q,\vec{r},g(q,b))$.
We will explain $\#\mathord{succ}$ as a balls-and-bins problem.
We need to assign a nonempty successor set to each of the $\vec{r}_{q,b}$ paths ending in $(q,b)$ such that for each $(q',b') \in Q\times\B$ there are $g(q,b)_{q',b'}$ sets containing $(q',b')$.
This is equivalent to dividing $\sum_{q',b'} g(q,b)_{q',b'}$ coloured balls, with colours in $Q \times \B$, over $\vec{r}_{q,b}$ bins in such a way that no bin has two balls of the same colour and no bin is empty.
We can express the number of distributions where bins may be empty by distributing the balls per colour independently, this is given by $\prod_{q',b'} {n \choose g(q,b)_{q',b'}}$ where $n$ is the number of bins (i.e. $\vec{r}_{q,b}$).
Using the inclusion-exclusion principle, then, we can add the restriction that no bins may be empty by including/excluding sets of empty bins.
This gives us the formula $\#\mathord{succ}(q,\vec{r},g(q,b)) = \sum_{j=0}^n (-1)^j{n \choose j}\left(\prod_{q',b'} {n-j \choose \vec{r'}[q,b]_{q',b'}}\right)$.

Hence, the number of elements in the equivalence class of $g$ is given by $\prod_{q,b}\#\mathord{succ}(q,\vec{r}_{q,b},g(q,b))$ and therefore we have that $w(\vec{r},a,\vec{r'}) = |C'| = \sum_{g \in \mathcal{G}} \prod_{q,b}\#\mathord{succ}(q,\vec{r}_{q,b},g(q,b))$.
\qed
\end{proof}
Using this lemma,, we can enumerate every vector in $C_{q,b}$ for every $(q,b)$-pair, which is a polynomial combination of polynomially sized objects. We can then calculate $\#\mathord{succ}$ with a PSPACE transducer to find $w(\vec{r},a,\vec{r'})$. Hence, we can construct $\kdis(\kA)$ with a PSPACE transducer.
\qed
\end{proof}

By Lemma \ref{lem:exponlower}, we have an $O(2^n)$ lower bound even if $k = n$. From Theorem \ref{thm:pspacetrans} we already have a $2^{O(n\log n)}$ upper bound, leaving only a small gap.

When the ambiguity is comparatively low, we can do even better:
\begin{theorem}\label{thm:polylogspacetrans}
Given a $k$-ambiguous automaton $\kA$ with $n$ states where $k = O(\log n)$, the disambiguation $\kdis(\kA)$ has at most $(2n)^{O(\log n)}$ states. Moreover, $\kdis(\kA)$ can be calculated using a POLYLOGSPACE transducer.
\end{theorem}
\begin{proof}
While in general, elements in $[\log n]^{[n] \times \Bin}$ take $2n\log n$ space to represent, we know that for any $\vec{r} \in Q''$, $\sum_i \vec{r}_i \leq \log n$ and hence we can instead represent elements in $Q''$ as vectors in $([n]\times\Bin)^{\log n}$, which take $\log(2n^{\log n}) = O((\log n)^2)$ space. Moreover, $w(\vec{r},a,\vec{r'})$ is at most exponential in the ambiguity of $\kA$ and hence we can simply enumerate the successors of any $P \in Q'$ with $\pi_{\mathord{last}}(P) = \vec{r}$ in POLYLOGSPACE. Thus, we can calculate $\kdis(\kA)$ in POLYLOGSPACE.
\qed
\end{proof}

%

\section{Model Checking IBA}\label{sub-mc-iba}

In this section we will consider the problem of model checking IBAs against Markov chains.
A \emph{Markov chain} (MC) is a pair $(S,P)$ where $S$ is the finite state set, and $P \in [0,1]^{S \times S}$ is a stochastic matrix specifying transition probabilities.
Given an initial distribution $\iota$, an MC $\mathcal{M}$ induces a probability measure~$\Pr^{\mathcal{M}}_\iota$ over infinite words.
The model checking question asks, what is the probability of the language accepted by an IBA?

We show that model checking IBAs against MCs can be done in NC using a modified procedure for model checking UBAs from~\cite{cav16full}. This is the main theorem:

\begin{theorem}\label{thm:mciba}
Let $\mathcal{M}$ be an MC and let $\A$ be an IBA. The probability that a random word sampled from $\mathcal{M}$ is in $L_\A$ can be computed in NC.
\end{theorem}

In order to prove this, we will use the following properties of nonnegative matrices:

\begin{theorem}\label{lem:matrixprop}
Let $M \in \R^{S\times S}$ be a nonnegative matrix. Then the following all hold:
\begin{enumerate}
	\item The spectral radius $\rho(M)$ is an eigenvalue of $A$ and there is a nonnegative eigenvector $\vec{x}$ with $M\vec{x} = \rho(M)\vec{x}$. Such a vector $\vec{x}$ is called \emph{dominant}.
	\item If $0 \leq M'\leq M$ then $\rho(M') \leq \rho(M)$.
	\item There is $C \subseteq S$ such that $M_{C,C}$ is strongly connected and $\rho(M_{C,C}) = \rho(M)$.
\end{enumerate}
\end{theorem}

\begin{theorem}\label{lem:perron}
Let $M \in \R^{S\times S}$ be a strongly connected nonnegative matrix. Then the following all hold:
\begin{enumerate}
	\item There is an eigenvector $\vec{x}$ with $M\vec{x} = \rho(M)\vec{x}$ such that $\vec{x}$ is strictly positive in all its components.
	\item The eigenspace associated with $\rho(M)$ is one-dimensional.
	\item If $0 \leq M' < M$ (i.e. $M'_{i,j} < M_{i,j}$ for some $i,j \in S$) then $\rho(M') < \rho(M)$.
	\item If $\vec{x} \geq 0$ and $M\vec{x} \leq \rho(M)\vec{x}$ then $M\vec{x} = \rho(M)\vec{x}$.
\end{enumerate}
\end{theorem}
These results can all be found in \cite[Chapter 2]{book:BermanP}.
We will also need the following well-known fact about Markov chains.
\begin{lemma}\label{lem:markovprefix}
Let $\mathcal{M} = (S,P)$ be a Markov chain, and $\mathcal{L} \subseteq S^\omega$ an $\omega$-regular language. Suppose $s_0 \in S$ such that $\prob_{s_0} > 0$. Then there exists $s_0\ldots s_n \in \paths_{s_0}(P)$ such that $\prob_{s_n}(\{w \in s_nS^\omega : s_0\ldots s_{n-1}w \in \mathcal{L}\}) = 1$.
\end{lemma}
This can be found for instance in \cite[Lemma 16]{cav16full}.

We will now outline the model checking procedure. As in \cite{cav16full} the algorithm will consist of solving a system of linear equations, split up in two parts. The first part is the basic linear system, which we will describe below, while the second part will add normalising equations in the form of cuts (or pseudo-cuts, if \cite{kiefer2019} is followed). Since the weights of the edges within SCCs are all on loops, and therefore all have weight 1, this part can be copied verbatim from Baier et al\cite{cav16full}. The difference is in the basic system of equations.

Given an automaton $\mathcal{A} = (Q, \Sigma, M, \alpha, F)$ and Markov chain $\mathcal{M} = (S, P)$, we will write $\mathcal{A}[q]$ (resp. $\mathcal{M}[s]$) to denote the weighted automaton (resp. Markov chain) starting in $q \in Q$ (resp. $s \in S$). W.l.o.g.~we will assume that every $q \in Q$ is reachable from some $q'$ with $\alpha(q') \neq 0$ and every $q \in Q$ can reach a loop over a final state. Note that while $\mathcal{A}$ is an IBA, $\mathcal{A}[q]$ does not necessarily have to be. We will write $\expect_{\mathcal{M}} L_{\mathcal{A}}$ for the expected weight of a word in $\mathcal{A}$ for a random word generated by $\mathcal{M}$. Since the weight of a word in an IBA is  either 0 or 1, we see that $\expect_{\mathcal{M}} L_{\mathcal{A}} = \prob_\mathcal{M}( L_{\mathcal{A}})$ if $\mathcal{A}$ is an IBA.

\begin{lemma}\label{lem:factorexpect}
Let $\mathcal{A} = (Q, \Sigma, M, \alpha, F)$ be an IBA w.r.t. Markov chain $\mathcal{M} = (S, P)$ with initial distribution $\iota$. The following equations hold:
\begin{eqnarray}
\prob_{\mathcal{M}} (L_{\mathcal{A}}) & = & \sum_{s \in S} \iota(s) \sum_{q \in Q} \alpha(q)\expect_{\mathcal{M}[s]} L_{\mathcal{A}[q]} \\
\expect_{\mathcal{M}[s]} L_{\mathcal{A}[q]} & = & \sum_{s' \in S} \sum_{q' \in Q} P_{s,s'}M(s)_{q,q'}\expect_{\mathcal{M}[s']} L_{\mathcal{A}[q']}
\end{eqnarray}
\end{lemma}
\begin{proof}
The first equation holds because of the fact that $\mathcal{A}$ is an IBA. Note that for any $q \in Q$, $|\mathit{FinalPaths}_{\mathcal{A}[q]}(w)| \leq N$ for all $w \in S^\omega$. For the second equation we have the following:
\begin{eqnarray*}
&& \expect_{\mathcal{M}[s]} L_{\mathcal{A}[q]}\\
& = & \int_{w \in \paths(\mathcal{M}[s])} L_{\mathcal{A}[q]}(w)d\prob(w) \\
& = & \sum_{s' \in S} \int_{w \in \paths(\mathcal{M}[s]) \cap s'S^\omega} L_{\mathcal{A}[q]}(w)d\prob(w) \quad \textrm{since the cylinder sets $sS^\omega$ partition $S^\omega$} \\
& = & \sum_{s' \in S} P_{s,s'}\int_{w \in \paths(\mathcal{M}[s'])} L_{\mathcal{A}[q]}(sw)d\prob(w) \quad \textrm{by definition of $\prob(w)$} \\
& = & \sum_{s' \in S} P_{s,s'} \int_{w \in \paths(\mathcal{M}[s'])} \sum_{p \in \mathit{FinalPaths}_{\mathcal{A}[q]}(sw)} \mathit{weight}(p)d\prob(w) \quad \textrm{definition of $L_{\mathcal{A}[q]}$} \\
& = & \sum_{s' \in S} P_{s,s'} \int_{w \in \paths(\mathcal{M}[s'])} \sum_{q' \in Q} M(s)_{q,q'}\sum_{p \in \mathit{FinalPaths}_{\mathcal{A}[q']}(w)} \mathit{weight}(p)d\prob(w) \\ && \hspace{75mm} \textrm{definition of $\mathit{FinalPaths}$} \\
& = & \sum_{s' \in S} \sum_{q' \in Q}P_{s,s'} M(s)_{q,q'}  \int_{w \in \paths(\mathcal{M}[s'])}\sum_{p \in \mathit{FinalPaths}_{\mathcal{A}[q']}(w)} \mathit{weight}(p)d\prob(w) \\
& = & \sum_{s' \in S} \sum_{q' \in Q}P_{s,s'} M(s)_{q,q'}\int_{w \in \paths(\mathcal{M}[s'])} L_{\mathcal{A}[q']}(w)d\prob(w) \quad \textrm{by definition of $L$} \\
& = &  \sum_{s' \in S} \sum_{q' \in Q}P_{s,s'} M(s)_{q,q'} \expect_{\mathcal{M}[s']}L_{\mathcal{A}[q']}
\end{eqnarray*}
\qed
\end{proof}

Let $\mathcal{M} = (S, P)$ be a Markov chain with state set $S$, and let $\iota$ be an initial distribution on $S$. Let $B \in \R^{(Q \times S) \times (Q \times S)}$ be the following matrix:
\begin{equation}
B_{\langle q,s\rangle,\langle q',s'\rangle} = P_{s,s'}M(s)_{q,q'}.
\end{equation}
Define $\vec{z} \in \R^{Q \times S}$ by $\vec{z}_{q,s}= \expect_{\mathcal{M}[s]} L_{\mathcal{A}[q]}$, i.e., the expected weight of a word for the Markov chain starting in $s$ and the automaton starting in $q$. Similar to \cite[Lemma 4]{cav16full} we have the following:

\begin{lemma}\label{lem:syseq}
Let $B$ and $\vec{z}$ as defined above. We have that $\vec{z} = B\vec{z}$.
\end{lemma}
\begin{proof}[of Lemma~\ref{lem:syseq}]
Using Lemma \ref{lem:factorexpect}:
\begin{align*}
\vec{z}_{q,s} \ & =\  \expect_{\mathcal{M}[s]} L_{\mathcal{A}[q]} \\
& =\  \sum_{s' \in S} \sum_{q' \in Q} P_{s,s'}M(s)_{q,q'} \expect_{\mathcal{M}[s']} L_{\mathcal{A}[q']} \\
& =\  (B\vec{z})_{q,s} \tag*{\qed}
\end{align*}
\end{proof}
W.l.o.g.\ we can assume that for every $\langle qs \rangle$ in $B$ we can reach an accepting loop. We need to show that SCCs in $B$ have a spectral radius of at most 1. In fact, we can give a probabilistic interpretation to values in $(B_{C,C})^n$, where $C \subseteq Q\times S$ is an SCC in $B$. This mirrors Proposition 6 in \cite{cav16full}.
\begin{lemma}\label{lem:powersinterpret}
Let $C \subseteq Q \times S$ and $\langle qs \rangle, \langle rt \rangle \in C$. Let $n \in \N$. Define $A := B_{C,C}$.
Then $(A^n)_{\langle qs\rangle, \langle rt\rangle} = \expect_{\mathcal{M}[s]} \mathit{weight}^{C,n}_{\langle qs \rangle,\langle rt \rangle}(w)$, where
\begin{equation*}
\mathit{weight}^{C,n}_{\langle qs \rangle,\langle rt \rangle}(w) = \sum_{\langle q_0 s_0\rangle\ldots\langle q_ns_n\rangle \in \paths_{\langle qs \rangle,\langle rt \rangle}(A), w \in s_0\ldots s_nS^\omega}\prod_{0\leq i< n} M(s_i)_{q_i,q_{i+1}}.
\end{equation*}
(Intuitively, $\mathit{weight}^{C,n}_{\langle qs \rangle, \langle rt \rangle}(w)$ is the weight of the prefixes of length $n$ of paths over $w$ in $C$ starting in $\langle qs \rangle$ that reach $\langle rt \rangle$ at time $n$.)

Define
\begin{equation*}
E^{C, n}_{\langle qs \rangle,\langle rt \rangle} := \{ s_0s_1\ldots \in sS^\omega \mid \exists q_1\ldots q_n . \langle q_0 s_0\rangle\ldots\langle q_ns_n\rangle \in \paths_{\langle qs \rangle,\langle rt \rangle}(A)\}.
\end{equation*}

In particular, if $C$ is (a subset of) an SCC, then $(A^n)_{\langle qs\rangle, \langle rt\rangle} = \prob_{\mathcal{M}[s]}\left(E^{C, n}_{\langle qs \rangle,\langle rt \rangle}\right)$ and hence $\rho(A) \leq 1$.
\end{lemma}
\begin{proof}
\begin{eqnarray*}
&& \expect_{\mathcal{M}[s]} \mathit{weight}^{C,n}_{\langle qs \rangle,\langle rt \rangle}(w) \\
& = & \int_{w \in \paths(\mathcal{M}[s_0])}\mathit{weight}^{C,n}_{\langle qs \rangle, \langle rt \rangle}(w)d\prob(w) \\
& = & \sum_{s_0\ldots s_n \in S^n}\int_{w \in \paths(\mathcal{M}[s])\cap s_0\ldots s_nS^\omega}\mathit{weight}^{C,n}_{\langle qs \rangle, \langle rt \rangle}(w)d\prob(w) \\
& = & \sum_{s_0\ldots s_n \in S^n}\int_{w \in \paths(\mathcal{M}[s_n])}\prod_{0 \leq i < n}P_{s_i, s_{i+1}}\mathit{weight}^{C,n}_{\langle qs \rangle, \langle rt \rangle}(s_0\ldots s_nw)d\prob(w) \\
& = & \sum_{s_0\ldots s_n \in S^n}\int_{w \in \paths(\mathcal{M}[s_n])}\prod_{0 \leq i < n}P_{s_i, s_{i+1}} \\
&  &\sum_{\langle q_0 s_0\rangle\ldots\langle q_ns_n\rangle \in \paths_{\langle qs \rangle,\langle rt \rangle}(A), w \in s_0\ldots s_nS^\omega}\prod_{0\leq i< n} M(s_i)_{q_i,q_{i+1}}d\prob(w) \\
& = & \sum_{s_0\ldots s_n \in S^n}\int_{w \in \paths(\mathcal{M}[s_n])} \\
&  &\sum_{\langle q_0 s_0\rangle\ldots\langle q_ns_n\rangle \in \paths_{\langle qs \rangle,\langle rt \rangle}(A)), w \in s_0\ldots s_nS^\omega}\prod_{0\leq i< n}P_{s_i, s_{i+1}} M(s_i)_{q_i,q_{i+1}}d\prob(w) \\
& = & \sum_{s_0\ldots s_n \in S^n}\sum_{\langle q_0 s_0\rangle\ldots\langle q_ns_n\rangle \in \paths_{\langle qs \rangle,\langle rt \rangle}(A)}\prod_{0\leq i< n}P_{s_i, s_{i+1}} M(s_i)_{q_i,q_{i+1}} \\
& = & \sum_{s_0\ldots s_n \in S^n}\sum_{\langle q_0 s_0\rangle\ldots\langle q_ns_n\rangle \in \paths_{\langle qs \rangle,\langle rt \rangle}(A)}\prod_{0\leq i< n} A_{\langle q_is_i\rangle,\langle q_{i+1}s_{i+1}\rangle} \\
& = & (A^n)_{\langle qs \rangle, \langle rt \rangle}
\end{eqnarray*}
Note that since an accepting loop can be reached from every state in $B$, no SCC in $B$ can have diamonds: otherwise, one could traverse a loop over that diamond $\log N + 1$ times to get a word that has more than $N$ final paths. Since any path that stays in an SCC has weight 1, and there are no diamonds in SCCs, we see that if $C$ is (a subset of) an SCC, then $\expect_{\mathcal{M}[s]} \mathit{weight}^{C,n}_{\langle qs \rangle,\langle rt \rangle}(w) = \prob_{\mathcal{M}[s]}\left(E^{C, n}_{\langle qs \rangle,\langle rt \rangle}\right)$ and hence $(A^n)_{\langle qs\rangle, \langle rt\rangle} = \prob_{\mathcal{M}[s]}\left(E^{C, n}_{\langle qs \rangle,\langle rt \rangle}\right)$.
\qed
\end{proof}

The definitions of recurrent SCCs and cuts are completely equivalent to those in \cite{cav16full} and are copied here verbatim: a \emph{recurrent SCC} is an SCC with spectral radius 1. We call a recurrent SCC $D$ \emph{accepting} if for some $\langle q t\rangle \in D$ we have $q \in F$. Let $D \subseteq Q \times S$ be an SCC of $B$. A set $\alpha \subseteq D$ is called a \emph{fiber} if it can be written as $\alpha = A \times \{s\}$ for some $A \subseteq Q, s \in S$. Given such a fiber $\alpha$ and $t \in S$, if $M_{s,t} > 0$ we then define a fiber
\begin{equation*}
\alpha \triangleright t := \{\langle q't\rangle \in D \mid q' \in \delta(q, s)\}
\end{equation*}
If $M_{s,t}=0$ then $\alpha \triangleright t$ is left undefined. We extend this definition inductively by writing $\alpha \triangleright \epsilon = \alpha$ and $\alpha \triangleright wt = (\alpha \triangleright w) \triangleright t$ for $t \in S, w \in S^*$. If $\alpha$ is a singleton $\{d\}$ we may write $d \triangleright w$ for $\alpha \triangleright w$.

We call a fiber $\alpha \subseteq D$ a \emph{cut} of $D$ if (i) $\alpha = d \triangleright v$ for some $d \in D$ and $v \in S^*$, and (ii) $\alpha \triangleright w \neq \emptyset$ holds for all $w \in S^*$ such that $\alpha \triangleright w$ is defined. Clearly if $\alpha$ is a cut then so is $\alpha \triangleright w$ when the latter is defined. Given a cut $\alpha \subseteq D$, we call its characteristic vector $\vec{\mu} \in \{0, 1\}^D$ a \emph{cut vector}.

Let $D$ be a recurrent SCC. A vector $\vec{\mu} \in [0,1]^D$ is called a \emph{$D$-normaliser} if $\vec{\mu}^T\vec{z} = 1$.

We will need versions of~\cite[Lemma 8]{cav16full} and~\cite[Lemma 10]{cav16full} as well. Lemma 10.1 only relies on the probabilistic interpretation of submatrices in SCCs as shown in Lemma~\ref{lem:powersinterpret} and the nonnegativity of those matrices, and hence we can invoke its proof verbatim:
\begin{lemma}[Lemma~10.1 in~\cite{cav16full}]\label{lem:cavlem10.1}
Let $D \subseteq Q \times S$ be an SCC. Then $D$ is recurrent if and only if it has a cut.
\end{lemma}
\begin{proof}
Verbatim from the proof in~\cite{cav16full}.
\qed
\end{proof}

For our proof of \cite[Lemma 8.1]{cav16full} we will first show something stronger:
\begin{lemma}\label{lem:nonrecurreach}
Let $D$ be a recurrent SCC, and let $C \subseteq Q \times S \setminus D$ be the set of states outside of $D$ reachable from $D$. Then $\rho(B_{C,C}) < 1$.
\end{lemma}
\begin{proof}
Towards a contradiction, assume $\rho(B_{C,C}) = 1$. Then by Theorem~\ref{lem:matrixprop}.3 there exists a recurrent SCC $D'$ such that $D'$ is reachable from $D$. We will create a word consisting of five parts such that $\A$ has more than $N$ final paths over that word. Since $N$ is the global bound on the number of final paths in $\A$, this leads to a contradiction.

\begin{enumerate}
	\item Firstly, pick $d = \langle qs_0 \rangle \in D$. Let $s_1\ldots s_i$ be such that $d \triangleright s_1\ldots s_i$ is a cut, such a word exists by Lemma~\ref{lem:cavlem10.1}.
	\item Let $s_{i+1}\ldots s_j$ be such that there exists a path from $d$ to some $d' \in D'$ over $s_1\ldots s_j$.
	\item Let $s_{j+1}\ldots s_k$ be such that $d' \triangleright s_{j+1}\ldots s_k$ is a cut. Again this exists by Lemma~\ref{lem:cavlem10.1}.
	\item Let $s_{k+1}\ldots s_l$ be such that $d \in d \triangleright s_1\ldots s_l$. This exists because $d \triangleright s_1\ldots s_i$ is a cut and for any cut $c$ and any $s \in S$ we have that $c \triangleright s$ is also a cut.
\end{enumerate}
Then consider the word $(s_1\ldots s_l)^{|D'|(N+1)}$. Because $d \in d \triangleright s_1\ldots s_l$ and we reach a cut in $D'$ from $d$ over $s_1\ldots s_l$, this means at least $|D'|(N+1)$ cuts in $D'$ are reached after reading $(s_1\ldots s_l)^{|D'|(N+1)}$ from $D$. By the pigeonhole principle, there exists a $d' \in D'$ such that there are at least $N+1$ paths from $d$ to $d'$ over  $(s_1\ldots s_l)^{|D'|(N+1)}$. Hence, if we pick $s_{l+1}\ldots$ such that there is a final path from $d'$ over $s_{l+1}\ldots$, we see that there are more than $N$ final paths from $d$ over $(s_1\ldots s_l)^{|D'|(N+1)}s_{l+1}\ldots$, contradicting the global bound on the number of final paths in $\A$. Thus, $\rho(D') < 1$ and therefore $\rho(B_{C,C}) < 1$.
\qed
\end{proof}

Now we can show Lemma~8.1:
\begin{lemma}[Lemma~8.1 in~\cite{cav16full}]\label{lem:cavlem8.1}
Let $D$ be a recurrent SCC, and let $\vec{x}$ be any vector satisfying $\vec{x}=B\vec{x}$. Then $\vec{x}_D = B_{D,D}\vec{x}_D$.
\end{lemma}
\begin{proof}
Let $\bar{D} = Q \times S \setminus D$ and let $C \subseteq \bar{D}$ be the set of states reachable from $D$. For all $n$ we have:
\begin{eqnarray*}
\vec{x}_D & = & (B^n\vec{x})_D \\
& = & B^{n-1}B_{D,D}\vec{x}_D + B^{n-1}B_{D,\bar{D}}\vec{x}_{\bar{D}} \\
& = & B^{n-1}B_{D,D}\vec{x}_D + B_{C,C}^{n-1}B_{D,\bar{D}}\vec{x}_{\bar{D}}.
\end{eqnarray*}
Since $\rho(B_{C,C}) < 1$ by Lemma~\ref{lem:nonrecurreach}, this second term goes to 0 as $n$ goes to infinity. Hence, $\vec{x}_D = B_{D,D}\vec{x}_D$.
\qed
\end{proof}

The proof of Lemma~8.2 in~\cite{cav16full} relies on the following lemma:
\begin{lemma}\label{lem:recursinks}
Let $D$ be a recurrent SCC in $B$, and let $\langle q, s\rangle \in D$. Then almost surely all final paths of $\mathcal{A}[q]$ are contained in $D$, and for any $\langle q', s' \rangle \not \in D$ reachable from $\langle q, s \rangle$, $z_{\langle q', s'\rangle} = 0$.
\end{lemma}
\begin{proof}
Suppose that with positive probability a word starting in $s$ is generated such that $\mathcal{A}[q]$ has a final path outside $D$. We will construct a word consisting of three repeating parts that admits more than $N$ final paths:
\begin{enumerate}
	\item Let $s_0\ldots s_l \in S^*$ be such that $s_0 = s$ and $\langle q,s\rangle \triangleright s_1\ldots s_l$ is a cut containing $\langle q,s\rangle$. Such a word exists by Lemma \ref{lem:cavlem10.1}.
	\item The language of words that have a final path outside $D$ is regular. The probability of generating a word with a final path outside $D$ is bounded from above by the sum of probabilities of reaching some $\langle q', s' \rangle$ outside $D$ times the probability of generating a word starting in $s'$ over which there exists a final path from $q'$, and the probability of generating a word with a final path outside $D$ is nonzero. Therefore, there must exist $\langle q', s' \rangle \not \in D$ reachable from $\langle q, s\rangle$ such that with nonzero probability, a word starting in $s'$ is generated over which $\A[q']$ has a final path. Let $s_l \ldots s_{m'}$ be such that $s_{m'} = s'$ and $M(s_l\ldots s_{m'})_{q,q'} > 0$. By Lemma \ref{lem:markovprefix} there exists a word $s_{m'} \ldots s_m$ such that $\prob_{s_m}(\{w \in s_mS^\omega \mid \textrm{there exists a final path from $q'$ over $s_{m'}\ldots s_mw$}\})=1$. Hence, $s_l\ldots s_m$ is such that $\prob_{s_l}(s_l\ldots s_mS^\omega) > 0$ and any word prefixed by $s_l\ldots s_m$ almost surely has a final path outside $D$.
	\item Let $\langle q'',s''\rangle \in \langle q,s\rangle \triangleright s_0\ldots s_m$ (which exists since $\langle q,s \rangle \triangleright s_0\ldots s_l$ is a cut and by extension so is $\langle q,s\rangle \triangleright s_0\ldots s_m$). By strongly connectedness of $D$ there exists $s_m\ldots s_n \in S^*$ such that $s_m = s''$ and $\langle q,s\rangle \in \langle q'',s''\rangle \triangleright s_{m+1}\ldots s_n$.
\end{enumerate}
Now consider the word $(s_0\ldots s_n)^{N+1}$. By construction, $\langle q,s\rangle \in \langle q,s \rangle \triangleright s_0\ldots s_n$ and there exists $\langle q', s'\rangle \not \in D$ such that $q$ reaches $q'$ over $s_0\ldots s_{m-1}$ and $\prob_{s_m}(\{w \in s_mS^\omega \mid \textrm{there exists a final path from $q'$ over $s_{m'}\ldots s_mw$}\})=1$. Hence, after every iteration of $s_0\ldots s_n$, the automaton can nondeterministically choose to stay in $D$ or move to a set of states out of $D$ after which it will almost surely have a final path. Thus, for almost any word prefixed with $(s_0\ldots s_n)^{N+1}$, there are at least $N+1$ final paths, which contradicts the fact that $\A$ is an IBA. This shows that almost surely any final path from $\langle q, s \rangle \in D$ stays in $D$, and hence for any $\langle q', s'\rangle \not \in D$ reachable from $\langle q, s\rangle$, $z_{\langle q', s'\rangle} = 0$.
\qed
\end{proof}

With this we can invoke the proof of Lemma~8.2 verbatim:
\begin{lemma}[Lemma~8.2 in~\cite{cav16full}]\label{lem:cavlem8.2}
Let $D$ be a recurrent SCC. For all $d \in D$, we have $\vec{z}_d > 0$ iff $D$ is accepting.
\end{lemma}
\begin{proof}
Verbatim from~\cite{cav16full}, relying only on Lemma~\ref{lem:recursinks} and the probabilistic interpretation in Lemma~\ref{lem:powersinterpret}.
\qed
\end{proof}

The proof of Lemma~10.2 in~\cite{cav16full} only relies on Lemma~\ref{lem:cavlem8.2} and hence can be invoked verbatim. This shows that cut vectors are $D$-normalisers.
\begin{lemma}[Lemma~10.2 in~\cite{cav16full}]\label{lem:cavlem10.2}
Let $D$ be an accepting recurrent SCC and let $\vec{\mu}$ be a cut vector in $D$. Then $\vec{\mu}^\top\vec{z}_D = 1$.
\end{lemma}
\begin{proof}
Verbatim.
\qed
\end{proof}

Finally we can prove our version of Lemma~12:
\begin{lemma}[Lemma~12 in~\cite{cav16full}]\label{lem:cavlem12}
Let $\mathcal{D}_+$ be the set of accepting recurrent SCCs, and $\mathcal{D}_0$ the set of non-accepting recurrent SCCs. For each $D \in \mathcal{D}_+$, let $\vec{\mu}_D \in [0,1]^D$ be a $D$-normalizer (which exists by Lemma~\ref{lem:cavlem10.2}). Then $\vec{z}$ is the unique solution of the following linear system:
\begin{eqnarray*}
& \vec{\zeta} & = B\vec{\zeta} \\
\textrm{for all $D \in \mathcal{D}_+$:} & \vec{\mu}^\top_D\vec{\zeta}_D & =  1 \\
\textrm{for all $D \in \mathcal{D}_0$:} & \vec{\zeta}_D & =  \vec{0}
\end{eqnarray*}
\end{lemma}
\begin{proof}
This proof is equivalent to the one in~\cite{cav16full} and is included for completeness.
The vector $\vec{z}$ solves the system of equations by Lemma~\ref{lem:syseq}, Lemma~\ref{lem:cavlem8.2}, and the definition of a $D$-normalizer.
To show uniqueness, let $\vec{x}$ solve the system of equations. We show that $\vec{x}=\vec{z}$.
\begin{itemize}
	\item Let $D \in \mathcal{D}_0$. Then $\vec{x}_D = \vec{0} = \vec{z}_D$.
	\item Let $D \in \mathcal{D}_+$. By Lemma~\ref{lem:cavlem8.1}, we have $\vec{x}_D = B_{D,D}\vec{x}_D$ and $\vec{z}_D = B_{D,D}\vec{z}_D$. By Theorem~\ref{lem:perron}, the eigenspace of $B_{D,D}$ is one-dimensional, implying that $\vec{x}_D$ is a scalar multiple of $\vec{z}_D$. We have $\vec{\mu}_D^\top\vec{x}_D = 1 = \vec{\mu}_D^\top\vec{z}_D$, hence $\vec{x}_D = \vec{z}_D$.
	\item Let $D := \bigcup \mathcal{D}_+ \cup \bigcup \mathcal{D}_0$ be the union of all recurrent SCCs, and let $\bar{D} := (Q \times S) \setminus D$. We have $\vec{x}-\vec{z} = B(\vec{x}-\vec{z})$. It follows
\begin{eqnarray*}
\vec{x}_{\bar{D}}-\vec{z}_{\bar{D}} & = & B_{\bar{D},\bar{D}}(\vec{x}_{\bar{D}}-\vec{z}_{\bar{D}}) + B_{\bar{D},D}(\vec{x}_D-\vec{z}_D) \\
& = & B_{\bar{D},\bar{D}}(\vec{x}_{\bar{D}}-\vec{z}_{\bar{D}}) + B_{\bar{D},D}(\vec{0}) \\
& = & B_{\bar{D},\bar{D}}(\vec{x}_{\bar{D}}-\vec{z}_{\bar{D}})
\end{eqnarray*}
		by the previous two items. By Lemma~\ref{lem:nonrecurreach}, $\rho(B_{\bar{D},\bar{D}}) < 1$. Thus, $\vec{x}_{\bar{D}} = \vec{z}_{\bar{D}}$. \qed
\end{itemize}
\end{proof}

This enables us to prove Theorem~\ref{thm:mciba}:
\begin{proof}[of Theorem~\ref{thm:mciba}]
It suffices to calculate $D$-normalisers and solve the system of equations described in Lemma~\ref{lem:cavlem12}. Solving a system of equations can be done in NC, and we can find a $D$-normaliser using~\cite[Lemma 22]{cav16full} in NC, proving the theorem.
\qed
\end{proof}

\begin{corollary}\label{cor:polylogspacekaba}
Let $\mathcal{M}$ be an MC and let $\kA$ be a $k$-ABA with $n$ states, where $k = O(\log n)$. The probability that a random word sampled from $\mathcal{M}$ is in $L_\A$ can be computed in POLYLOGSPACE.
\end{corollary}
\begin{proof}
This combines~\Cref{thm:polylogspacetrans,thm:mciba}. Since NC is contained in POLYLOGSPACE~\cite{Pap94}, this proves the corollary.
\qed
\end{proof}

PSPACE-hardness of model checking $k$-ABAs against Markov chains can be shown easily from the \emph{finite intersection problem} introduced by D.~Kozen~\cite{kozen1977}.
This problem asks, given a set of deterministic automata, whether the intersection of the languages is empty. By considering the union of the complements of the automata as a $k$-ABA we see that the probabilistic universality problem for $k$-ABAs is PSPACE-hard.
Membership in PSPACE is shown by translating the $k$-ABA into an equivalent IBA and model checking this IBA against the Markov chain.
\begin{theorem}
The model checking problem for $k$-ABAs is PSPACE-complete.
\end{theorem}

\begin{corollary}
Let $\mathcal{M}$ be an MC and let $\A$ be an NBA with $n$ states. The probability that a word in $L_\A$ is accepted by $\mathcal{M}$ can be computed in PSPACE.
\end{corollary}
\begin{proof}
By Löding et al.~\cite{loding18}, $\A$ can be converted in a $k$-ABA $\kA$ with at most $3^n$ states where $k=n$. Hence, using Corollary~\ref{cor:polylogspacekaba}, we can calculate $\Pr(L_{\kA}) = \Pr(L_\A)$ in POLYLOG($2^n$) = POLY($n$) space.
\qed
\end{proof} 

\subsubsection*{Acknowledgements}
Mathieu Kubik, Tomasz Ponitka, and Joao Paulo Costalonga contributed ideas to the proofs of Lemmas \ref{lem-Z2-R}, \ref{lem-R-Q}, and \ref{lem:wrarcombi}, respectively.

\bibliographystyle{splncs04}
\bibliography{lit}


\end{document}